\documentclass[a4paper,11pt]{article}
\usepackage{amsmath,amssymb,amscd,amsthm}
\usepackage{a4wide}
\usepackage{graphicx} 
\usepackage{tikz}

\usepackage{subcaption}

\newtheorem{theorem}{Theorem}[section]

\newtheorem{proposition}[theorem]{Proposition}
\newtheorem{lemma}[theorem]{Lemma}
\newtheorem{corollary}[theorem]{Corollary}

\newtheorem{definition}[theorem]{Definition}

\theoremstyle{definition}

\newtheorem{remark}[theorem]{Remark}

\begin{document}

\newcommand{\legendre}[2]{\genfrac{(}{)}{}{}{#1}{#2}} 

\newcommand{\beq}{\begin{equation}}  
\newcommand{\eeq}{\end{equation}}  
\newcommand{\bea}{\begin{eqnarray}}  
\newcommand{\eea}{\end{eqnarray}}  
\newcommand\la{{\lambda}}   
\newcommand\La{{\Lambda}}   
\newcommand\ka{{\kappa}}   
\newcommand\al{{\alpha}}   
\newcommand\be{{\beta}} 
\newcommand\ze{{\eta}} 
\newcommand\zet{{\nu}} 
\newcommand\gam{{\gamma}}     
\newcommand\om{{\omega}}  
\newcommand\tal{{\tilde{\alpha}}}  
\newcommand\tbe{{\tilde{\beta}}}   
\newcommand\tla{{\tilde{\lambda}}}  
\newcommand\tmu{{\tilde{\mu}}}  
\newcommand\si{{\sigma}}  
\newcommand\lax{{\bf L}}    
\newcommand\mma{{\bf M}}    
\newcommand\rd{{\mathrm{d}}}  
\newcommand\tJ{{\tilde{J}}}  
\newcommand\ri{{\mathrm{i}}} 
\newcommand\lcm{{\mathrm{lcm}}} 

\newcommand\rS{{\mathrm{S}}} 
\newcommand\rJ{{\mathrm{J}}} 

\newcommand{\F}{{\mathbb F}}
\newcommand{\N}{{\mathbb N}}
\newcommand{\Q}{{\mathbb Q}}
\newcommand{\Z}{{\mathbb Z}}
\newcommand{\C}{{\mathbb C}}
\newcommand{\R}{{\mathbb R}}

\newcommand\tI{{\tilde{\mathcal{I}}}}

\newcommand\SH{{\mathcal{S}}}

\newcommand\Tc{{\mathcal{T}}}

\newcommand\Uc{{\mathcal{U}}}

\newcommand\Pc{{\mathcal{P}}}

\newcommand\tr{{{\mathrm{tr}}\,}}


\title{Linear relations for Laurent  polynomials and lattice equations}
\author{Andrew N.W. Hone\footnote{Work partly carried out while on leave in the 
School of Mathematics and Statistics, 
University of New South Wales, Sydney NSW 2052, Australia. }$\,$
and Joe Pallister 
\\School of Mathematics, Statistics \& Actuarial Science \\ University of Kent, UK. 
}
\maketitle

\begin{abstract}
A recurrence relation is said to have the Laurent property if all of its iterates are Laurent 
polynomials in the initial values with integer coefficients.
We consider a family of nonlinear recurrences with the Laurent property, which were 
derived by Alman et al.\  via a construction of periodic seeds in Laurent phenomenon 
algebras, and generalize the Heideman-Hogan recurrences. Each member of the 
family is shown to be linearizable, in the sense that the iterates satisfy linear recurrence 
relations with constant coefficients. 
The latter are obtained from linear relations with periodic coefficients, which were 
found recently by Kamiya et al.\ from travelling wave reductions of a linearizable lattice equation on a 
6-point stencil. We introduce another linearizable  lattice equation on the same 
stencil, and present the corresponding linearization for its travelling wave reductions. 
Finally, for both of the 6-point lattice equations 
considered, we use the formalism of van der Kamp to construct a broad 
class of initial value problems with the Laurent property.  
\end{abstract}

\section{Introduction}


There continues to be a great deal of interest in nonlinear recurrences of the form 
\begin{equation}\label{mutationrecurrence}
x_{n+m}x_n=P(x_{n+1},\ldots x_{m+n-1}), 
\end{equation}
for a  polynomial $P$, 
with the surprising property that all  of the iterates are Laurent polynomials in the initial data with integer coefficients, that is to say 
$$
x_n\in \mathbb{Z}[x_0^{\pm 1},\ldots,x_{m-1}^{\pm 1}]
$$
for all $n$ \cite{gale}. This Laurent property is a 
central feature of the generators in cluster algebras, a novel class of commutative algebras introduced by 
Fomin and Zelevinsky \cite{clusteri}, which are defined by recursive relations of the same  form as (\ref{mutationrecurrence}) but  with the restriction that $P$ should  be a binomial expression 
of a specific kind. The same authors also considered a more general set of sufficient conditions which ensure that the above recurrence has the Laurent property, 
without requiring $P$ to be a binomial \cite{laurentphenomenon}. More recently, this led to the introduction of the broader framework of 
Laurent phenomenon  algebras \cite{lp}.

Cluster algebras are the focus of much activity due to their connections with diverse areas of mathematics and physics, 
ranging from Lie theory and  supersymmetric gauge theories 
to  Teichm\"uller theory
and dimer models \cite{eager, fst, gk}. The 
structure of a typical
cluster algebra may be very complicated, due to the complexity of  the recursive process, called mutation, that produces the 
generators. However, there are certain subclasses of cluster algebras that are associated with  discrete integrable 
systems of some kind, and often  these are the examples that are  of most interest in applications to other 
areas. Beyond the finite type cluster algebras, which have a taxonomy that coincides with the Cartan-Killing classification of semisimple Lie algebras and finite root systems \cite{clusterii}, 
and are associated with purely periodic dynamics, the next interesting subclass corresponds to discrete dynamical systems that admit linearization, in the sense that the variables 
satisfy linear recurrence relations with constant coefficients. The simplest example is the recurrence 
\begin{equation} \label{kro} 
x_{n}x_{n+2} = (x_{n+1})^2+1,
\end{equation}
which arises from mutations of the Kronecker quiver (an orientation of the affine $A_1^{(1)}$ diagram), for which 
the iterates 
satisfy the linear relation 
$$ 
x_{n+2} - Cx_{n+1} + x_n =0
$$ 
for all $n\in\Z$, 
where 
$$
C=\frac{x_n}{x_{n-1}}+\frac{x_n}{x_{n-1}}+\frac{1}{x_nx_{n-1}}
$$ 
is a first integral (independent of $n$). Linearizability was found for dynamics of cluster variables obtained from affine Dynkin quivers of type $A$ in \cite{fordymarsh}, 
for types $A$ and $D$ via frieze patterns in \cite{frises}, and in general for all affine types $ADE$ in  \cite{kellerscherotzke}. It has been further conjectured (and proved 
in certain cases) that 
linearizability holds for sequences of cluster variables obtained from mutation sequences obtained from box products $X\Box Y$ of a finite type Dynkin quiver $X$ and an affine Dynkin 
quiver $Y$ \cite{zamoint}. A previously known example is provided by Q-systems \cite{dk}, which arise from the Bethe ansatz for quantum integrable models,  and 
correspond to taking $X=A_n$ for $n$ arbitrary and $Y=A_1^{(1)}$, so (\ref{kro}) is included when $n=1$. 
The case of a product of a pair of affine quivers $X,Y$ is not linearizable, but is conjectured 
to be associated with systems that are integrable in the Liouville-Arnold sense \cite{gp}.  

In work by one of us  with Fordy \cite{fordyhone}, concerning cluster algebras obtained from quivers that are mutation-periodic with period 1, in the sense 
of \cite{fordymarsh}, we showed a further property of the affine type $A$ recurrences, 
specified by a pair of coprime positive integers $p,q$ as  
\begin{equation}\label{affa} 
x_nx_{n+p+q} =x_{n+p}x_{n+q}+1, 
\end{equation} 
namely that the iterates satisfy additional linear relations with periodic coefficients, of the form 
\begin{equation}\label{pera} 
 \begin{array}{rcl} 
x_{n+2q}-J_n x_{n+q} +x_n & =& 0,\quad J_{n+p}=J_n, \\ 
x_{n+2p}-K_n x_{n+p} +x_n & =& 0,\quad K_{n+q}=K_n.  \end{array} 
\end{equation} 
 Furthermore,  another family of linearizable recurrences from period 1 quivers was found, of the form 
\begin{equation}\label{2fri} 
x_nx_{n+2k} =x_{n+p}x_{n+q} 
+x_{n+k}, \qquad p+q=2k, 
\end{equation} 
which includes Dana Scott's recurrence \cite{gale} 
\begin{equation}\label{danascott} 
x_n x_{n+4} = x_{n+1}x_{n+3} + x_{n+2}, 
\end{equation} 
and these also admit linear relations with periodic coefficients, given by  
\begin{equation}\label{perfr} 
 \begin{array}{rcl} 
x_{n+3q} -J_{n+k}x_{n+2q}+J_n x_{n+q} -x_n & =& 0,\quad J_{n+p}=J_n, \\ 
x_{n+3p}-K_{n+k}x_{n+2p}+K_n x_{n+p} -x_n & =& 0,\quad K_{n+q}=K_n.  \end{array} 
\end{equation} 
Analogous linear relations with periodic coefficients for affine quivers of type $D$ and $E$, 
and associated Liouville integrable systems, appear in \cite{joede}.

In this paper we are concerned with linearizable recurrences that exhibit the Laurent property but go beyond the setting of cluster algebras. 
To begin with we will consider the family of recurrences 
\begin{equation}\label{firstlittlepi}
x_nx_{n+2k+l}=x_{n+2k}x_{n+l}+ax_{n+k}+ax_{n+k+l}, 
\end{equation}
with a fixed parameter $a$ and positive integers $k$ and $l$. These recurrences were named the ``Little Pi" 
family in \cite{laurentphenomenonsequences}, where they were shown to be generated by period $1$ seeds in 
the setting of Laurent phenomenon (LP) algebras. They extend the Heideman-Hogan recurrences 
\cite{heidemanhogan},  corresponding
to the case $l=1$, for which detailed features of the linearization were proved in \cite{heidemanhoganhone}. 
Thus our first aim here is to generalize the results of the latter work, and 
resolve some open 
conjectures from \cite{wardthesis}.
In particular, for (\ref{firstlittlepi}) we 
obtain 
the constant coefficient linear relation 
\begin{equation}\label{introlinear}
x_{n+6kl}-\mathcal{K}x_{n+4kl}+\mathcal{K}x_{n+2kl}-x_n=0
\end{equation}
when 
$2k$ and $l$ are coprime, and a counterpart relation
\begin{equation} \label{introbiglin} 
x_{n+6kl}-\mathcal{A}x_{n+5kl}+\mathcal{B}x_{n+4kl}-\mathcal{C}x_{n+3kl}+\mathcal{B}x_{n+2kl}-\mathcal{A}x_{n+kl}+x_n=0).
\end{equation}
if $\gcd(2k,l)=2$. (All other cases can be reduced to one of these.) 
In addition to the first integrals ($\mathcal K$ or $\mathcal{A}, \mathcal{B}, \mathcal{C}$) that 
appear as coefficients, we derive periodic quantities and associated 
linear relations with periodic coefficients. 

Ordinary difference equations can arise as 
reductions of two-dimensional lattice equations. For instance, the  affine type $A$ recurrences (\ref{affa}) 
are obtained from the 4-point equation 
\begin{equation}\label{frieze} 
\left|\begin{array}{cc} u_{s,t} & u_{s+1,t} \\ 
 u_{s,t+1} & u_{s+1,t+1} 
 \end{array}\right|=1
\end{equation}
for $(s,t)$ being coordinates on $\Z^2$ (or more generally, on a 
quadrilateral lattice), 
which is the relation for a frieze pattern \cite{frises}. To obtain (\ref{affa}), one should take 
the 
$(p,-q)$ travelling wave reduction
\begin{equation}\label{trav} 
u_{s,t}= x_n, \qquad n = ps+qt, 
\end{equation}
corresponding to a wave moving on the lattice 
with constant velocity 
$-q/p\in\mathbb{Q}$. 
Similarly, it was noted in  \cite{kkmt} that the 5-point lattice equation 
\begin{equation}\label{2frieze} 
\left|\begin{array}{cc} u_{s,t-1} & u_{s+1,t} \\ 
 u_{s-1,t} & u_{s,t+1} 
 \end{array}\right|=u_{s,t}, 
\end{equation}
which is the relation for a 2-frieze \cite{mgot}, 
reduces to (\ref{2fri}) by substituting in (\ref{trav}) with the replacement 
$p\to p-k$, $q\to k$, to obtain 
the $(p-k,-k)$ 
travelling  wave reduction.
The authors of  \cite{kkmt} also considered the 
Little Pi family (\ref{firstlittlepi}) 
as a reduction of the 6-point lattice equation 
\begin{equation}\label{2dlittlepiintro}
u_{s+1,t+2}u_{s,t}=u_{s+1,t}u_{s,t+2}+a(u_{s,t+1}+u_{s+1,t+1}). 
\end{equation} 
(Note that,  compared with \cite{kkmt}, we have switched the order of the independent variables and 
introduced the parameter $a$.) 
By obtaining linear relations for the above 
lattice equation, they deduced linear recurrences with periodic coefficients 
for its $(l,-k)$ travelling wave reduction  (\ref{firstlittlepi}) 
(cf.\  Proposition \ref{periodickernel} and Corollary \ref{periodlrelation} below). 
In addition, they proved the 
Laurent property for the lattice equation (\ref{2dlittlepiintro}), 
in the sense that for the initial value problem defined by  
\[
I=\{u_{s,0},u_{s,1}, u_{0,t}:s,t\in \mathbb{N}\}, 
\]
the iterates in the positive quadrant in $\Z^2$ are Laurent polynomials in the elements of this set. 
In this paper we  introduce a new 6-point lattice equation, given by 
\begin{equation}\label{lattice}
(u_{s+1,t+2}+u_{s+1,t}+a)u_{s,t+1}=(u_{s,t+2}+u_{s,t}+a)u_{s+1,t+1}
\end{equation}
and prove the Laurent property for both this and (\ref{2dlittlepiintro}) 
with a much broader set of initial values than just $I$. 
We further show that (\ref{lattice}) is linearizable, and this feature (as well as the Laurent 
property) extends to the family of $(l,-k)$ travelling wave reductions 
\begin{equation}\label{1drecurrence}
(x_{n+2k+l}+x_{n+l}+a)x_{n+k}=(x_{n+2k}+x_n+a)x_{n+k+l}. 
\end{equation} 
Our original motivation for introducing (\ref{lattice}) was the fact that, when $k=1$, 
the reduction (\ref{1drecurrence}) is the total difference of 
\begin{equation}\label{extreme} 
x_{n+l+1}x_n=x_{n+l}x_{n+1}+a\sum_{i=1}^{l}x_{n+i}+b, 
\end{equation}
where the arbitrary parameter $b$ is an integration constant. 
The latter family of recurrences was referred to as 
the ``Extreme polynomial" in  \cite{laurentphenomenonsequences}, where it was  obtained  from 
another set of period 1 seeds in LP algebras, and for $b=0$  it was independently found   in \cite{wardthesis}, 
where it was also shown to be linearizable  and  have the Laurent property (see \cite{honeward} for 
further details). However, the recurrences (\ref{1drecurrence}) lie  beyond the setting of LP algebras. 

All of the lattice equations described above fit into the framework of 
partial differential, differential-difference and partial difference equations 
described by Demskoi and Tran \cite{dt}, 
who considered the family of determinantal equations 
\begin{equation}\label{deteqn} 
 |M| = \mathrm{const}, 
\end{equation} 
where $|M|=\det (M)$  
 is the determinant of an $N\times N$ matrix $M$ of Casorati type, with entries specified by 
$$ 
M= ( u_{s+i-1,t+j-1})_{1\leq i,j\leq N}  
$$ 
(up to shifts of indices) in the lattice case, or with appropriate  modifications to  Wronskian type entries in 
the case of partial   differential/differential-difference equations. Equations of the form (\ref{deteqn}) 
are connected to 2D Toda lattices with appropriate 
boundary conditions, as well as Liouville's equation, and they are said to be Darboux integrable, meaning that they 
admit complete sets of first integrals that do not depend on  one or the other of the independent variables $s,t$.   
The $SL_2$ frieze relation (\ref{frieze}) is already in the form (\ref{deteqn}) with $N=2$, 
and has the consequence that the corresponding $3\times 3$ determinant vanishes, i.e.\ 
$$ 
\left| \begin{array}{ccc} u_{s,t}&u_{s,t+1}& u_{s,t+2} \\ 
u_{s+1,t}&u_{s+1,t+1}& u_{s+1,t+2} \\ 
u_{s+2,t}&u_{s+2,t+1}& u_{s+2,t+2} 
\end{array} \right|=0, 
$$ 
which follows by applying  the Dodgson condensation algorithm \cite{dodgson}, based on 
the Desnanot-Jacobi identity for matrix minors \cite{bressoud} 
(this is also referred to as Sylvester's identity in \cite{dt}), that is   
\begin{equation} \label{dc} 
|M|\, |M^{1N}_{1N}| = |M^1_1|\, |M^N_N| - |M^1_N|\, |M^N_1|
\end{equation} 
in which a superscript $i$ (subscript $j$) on a minor denotes 
that the $i$th row ($j$th column) is deleted. 
Using similar methods to \cite{fordyhone}, the right/left null vectors 
of the   $3\times 3$ matrix yield the linear relations 
\begin{equation}\label{jk} 
 \begin{array}{rcl} 
u_{s,t+2}-J u_{s,t+1} +u_{s,t} & =& 0,\quad \Delta_s J := J(s+1,t) - J(s,t)=0, \\ 
u_{s+2,t}-K u_{s+1,t} +u_{s,t} & =& 0,\quad \Delta_t K := K(s,t+1) - K(s,t)=0, \end{array} 
\end{equation} 
where the coefficients $J=J(t)$, $K=K(s)$ are first integrals of (\ref{frieze}) in the $s,t$ directions respectively, 
and the two linear relations in (\ref{jk}) reduce to those in (\ref{pera}) after imposing the travelling wave reduction (\ref{trav}). 
Similarly, applying Dodgson condensation with the 2-frieze relation (\ref{2frieze}) yields 
$$
\left| \begin{array}{ccc} 
u_{s,t-2}&u_{s+1,t-1}& u_{s+2,t} \\ 
u_{s-1,t-1}&u_{s,t}& u_{s+1,t+1} \\ 
u_{s-2,t}&u_{s-1,t+1}& u_{s,t+2} 
\end{array} \right|=1, 
$$ 
which is the relation for an $SL_3$ frieze on each of the sublattices  
obtained by restricting $s+t$ to have odd/even parity,  and can be 
put in the standard form (\ref{deteqn}) by a linear change of coordinates. 
A further application of (\ref{dc}) shows that the corresponding $4\times 4$ 
determinant vanishes for the 2-frieze relation, while 
in \cite{kkmt} it is shown that there is also a constant $3\times 3$ 
determinant and a vanishing $4\times 4$ 
determinant associated with 
(\ref{2dlittlepiintro}), and 
in the sequel we prove an analogous result for the new  lattice equation  
(\ref{lattice}). 
  
In the next section we give a very brief introduction to LP algebras, and explain how nonlinear recurrences 
of the form (\ref{mutationrecurrence}) can arise in that setting, giving full details for the particular case of the Little Pi family 
(\ref{firstlittlepi}). Section \ref{littlepilin} is devoted to  an 
independent derivation of the linear recurrences with periodic coefficients found  
for the Little Pi family in \cite{kkmt}, which  we then use  
in Section \ref{littlepiconst} to derive 
a constant coefficient  relation of order $6kl$, of the form (\ref{introlinear}) or (\ref{introbiglin}), 
for each pair of coprime positive integers $k,l$.
Section \ref{latticesection}
is concerned with the new 6-point lattice equation (\ref{lattice}), including the proof of linearization 
both for the lattice equation and all its travelling wave reductions (\ref{1drecurrence}). Finally, in Section 
\ref{boatsection} we show that the new lattice equation has the Laurent property for suitable band sets of initial values 
in $\Z^2$, 
of the kind described by van der Kamp in \cite{boat}, 
and we use this to infer the Laurent property for its  reductions (\ref{1drecurrence}). 
We finish by applying the same approach  
to show the Laurent property for the lattice equation  (\ref{2dlittlepiintro}) 
with band sets of initial values.

\section{Laurent phenomenon algebras and recurrence relations} \label{lpalgebras} 
There are various situations where birational transformations of the form (\ref{mutationrecurrence}) arise 
with more than two monomials on the right-hand side \cite{gsvgln, wilson}, and Laurent 
phenomenon (LP) algebras provide a general framework for such situations which goes 
beyond the setting of cluster algebras \cite{lp}. 
Like cluster algebras, LP algebras are constructed from collections of objects called clusters: for an 
LP algebra of rank $m$, a cluster is a set of $m$ independent quantities called cluster variables. 
A seed in an LP algebra consists of  a cluster together with $m$ polynomials in the cluster variables, 
called exchange polynomials. There is a process called mutation which allows new seeds to be produced, 
using the exchange polynomials. Certain conditions must be imposed on the exchange polynomials which ensure that 
the Laurent property is preserved under arbirary sequences of mutations, in the sense that 
all of the cluster variables so obtained are Laurent polynomials in the $m$ cluster variables from the initial seed. 
There is also a concept of periodic seeds \cite{laurentphenomenonsequences}, 
analogous to the concept for cluster variables that was 
introduced in  \cite{fordymarsh}, which allows  recurrence relations to be generated by particular sequences of mutations.

\subsection{Construction of an LP algebra}

\begin{definition}
	
A seed $(\mathbf{x},\mathbf{P})$ in an LP algebra of rank $m$ 
consists of a cluster 
$\mathbf{x}$, which is a   
collection of $m$ algebraically independent elements  (cluster variables) $x_i$, so $\mathbf{x}=\{x_1,\ldots x_m\}$, 
together with $m$ exchange polynomials $\mathbf{P}=\{P_1,\ldots P_m\}$. For each $i\in 1,\ldots m$ 
it is required that 
\begin{itemize}
\item $P_i$ is  irreducible in $\mathbb{Z}[x_1,\ldots ,x_m]$; 
\item $P_i$  does not contain the variable $x_i$.
\end{itemize}
\end{definition}


\begin{definition}
For each seed $(\mathbf{x},\mathbf{P})$ there is a 
mutation $\mu_k$  for each $k\in\{1,\ldots, m\}$, producing a new seed $\mu_k((\mathbf{x},\mathbf{P}))=(\mathbf{x}',\mathbf{P}')$. The process of mutation is defined in the following steps:

\begin{enumerate}  
\item
Define the exchange Laurent polynomials $\{\hat{P}_1,\ldots , \hat{P}_m\}\subset \mathbb{Z}[x_1^{\pm 1},\ldots , x_m^{\pm 1}]$ 
to be the unique polynomials such that 
\begin{itemize}
\item 
$\hat{P}_j=P_j \prod_{1\leq i\leq n, i\neq j} x_i^{a_i}$
for each $j$ and  $a_i\in \mathbb{Z}_{\leq 0}$ for each $i$; 
\item for $i\neq j$, 
\[
\hat{P}_i|_{x_j \leftarrow \frac{P_j}{x}}\in \mathbb{Z}[x_1^{\pm 1},\ldots, x_{j-1}^{\pm 1},x^{\pm 1}, x_{j+1}^{\pm 1},\ldots, x_m^{\pm 1}]
\]
and this polynomial is not divisible by $P_j$ in this ring.
\end{itemize}
\item 
The new cluster is $\mathbf{x}'=\mu_k({\bf x}) 
=\{x_1,\ldots, x'_k, \ldots , x_m\}$ where $x'_k:= \frac{\hat{P}_k}{x_k}$
\item 
Now define polynomials 
\[
G_j:=P_j|_{x_k\leftarrow \frac{\hat{P}_k|_{x_j\leftarrow 0}}{x'_k}}
\]
\item 
For each $j$, remove all common factors with $\hat{P}_k|_{x_j\leftarrow 0}$ from $G_j$ in the 
unique factorization domain 
$\mathbb{Z}[x_1,\ldots ,\widehat{x_k},\ldots ,\widehat{x_j},\ldots ,x_m]$, with the hats denoting 
omitted variables. Denote the polynomials obtained in this way by $H_j$. 
\item
The new exchange polynomials are $P'_j=H_jM_j$, where $M_j$ is the unique Laurent monomial in $\mathbb{Z}[x_1,\ldots,x_{k-1},x'_k,x_{k+1},\ldots , x_m]$ such that $P'_j$ is not divisible by any Laurent monomial in this ring.
\item
The new seed is $(\mathbf{x}',\mathbf{P}')=(\{x_1,\ldots,x'_k,\ldots, x_m\},\{P'_1,\ldots , P'_m\})$
\end{enumerate}
\end{definition}

\begin{definition} 
Two seeds are said to be mutation equivalent if one can be obtained from the other via a finite sequence of mutations. 
For a choice of initial  seed $(\mathbf{x},\mathbf{P})$,  the LP algebra $\mathcal{A}=\mathcal{A}(\mathbf{x},\mathbf{P})$ 
is the subalgebra of $\mathbb{Q}(x_1,\ldots,x_m)$ generated by all cluster variables in seeds that are mutation equivalent to the initial seed. Evidently this does not depend on the choice of initial seed.
\end{definition}

The somewhat convoluted construction of LP mutation ensures that the proof of the Laurent property of cluster algebras, via the caterpillar lemma in \cite{clusteri}, is still valid in the more general LP case.

\begin{theorem}[\cite{laurentphenomenon}, Theorem 5.1]\label{LaurentProperty}
Each of the cluster variables in the LP algebra is a Laurent polynomial in the cluster variables of an initial seed, 
i.e.\ for any seed $(\{x_1,\ldots x_m\},\mathbf{P})$ we have 
$\mathcal{A}\subset \mathbb{Z}[x_1^{\pm 1},\ldots x_m^{\pm 1}]$.
\end{theorem}

\subsection{Recurrence relations from period $1$ seeds}

Following \cite{laurentphenomenonsequences}, 
we now show how the notion of periodic seeds for cluster algebras, 
introduced in  \cite{fordymarsh}, may be generalized to LP algebras, in the special case where the period is $1$. 
Periodic seeds may be used to show that the iterates of certain recurrence relations correspond 
to mutations in an LP algebra, hence satisfying the conditions of Theorem \ref{LaurentProperty}. This proves the Laurent property for these recurrences. In the definition above, we considered unordered seeds, but when 
we consider recurrence relations it is helpful to fix an ordering. 

\begin{definition}[Period $1$ seed]
Let 
$(\mathbf{x},\mathbf{P})= 
\Big((x_1,\ldots,x_{m}),(P_1,\ldots,P_{m})\Big)$ be a seed 
(where the cluster variables and exchange polynomials are ordered according to 
their subscript), 
let  
$$ (\mathbf{x}',\mathbf{P}')=
\rho\circ \mu_1(\mathbf{x},\mathbf{P}) 
=\Big((x'_2,\ldots,x'_{m},x_1'),(P'_2,\ldots,P'_{m},P_1')\Big)
$$
be the seed obtained from it by  applying the mutation $\mu_1$ and then reordering the 
variables with a cyclic permutation $\rho$, and define $x_{m+1}=x_1'$. 
The seed $(\mathbf{x},\mathbf{P})$ is called periodic with period $1$ if 
\beq\label{condns} 
P'_1=\mathcal{S}P_{m} \qquad  and \qquad 
P'_i=\mathcal{S}P_{i-1}\quad  for \quad 2\leq i \leq m, 
\eeq  
where the shift operator $\mathcal{S}$ increases the subscripts on each of the $x_i$ appearing by one. 
\end{definition}


The importance of the above definition is due to the following result, which is Corollary 2.5 in 
\cite{laurentphenomenonsequences}, but for completeness we sketch the proof here. 

\begin{proposition}\label{period1seed}
If $(\mathbf{x},\mathbf{P})= 
\Big((x_1,\ldots,x_{m}),(P_1,\ldots,P_{m})\Big)$
is a period $1$ seed with $P_1=P(x_2,\ldots, x_m)$,  then the iterates of recurrence
\begin{equation}\label{genericrecurrence}
x_nx_{n+m}=P(x_{n+1},\ldots,x_{n+m-1})
\end{equation}
are Laurent polynomials $x_1,\ldots, x_m$ with integer coefficients. 
\end{proposition}
\begin{proof}
Since $P'_2=\mathcal{S}P_1$ we have $P'_2=P(x_3,\ldots x_{m+1})$.
Applying the 
mutation $\mu_1$ to the  period $1$ seed $(\mathbf{x},\mathbf{P})$ 
gives $$x_1'=\frac{P(x_2,\ldots,x_{m})}{x_1},$$ 
and  setting $x_{m+1}=x_1'$ agrees with the first iteration of the proposed 
recurrence (\ref{genericrecurrence}). After applying the cyclic permutation $\rho$ to reorder 
the variables and exchange polynomials,  the new seed is 
$$ 
\rho\circ \mu_1(\mathbf{x},\mathbf{P})  =
\Big((x_2,\ldots,x_{m},x_{m+1}),(P'_2,\ldots,P'_{m},P_1')\Big), 
$$ 
with new exchange polynomials given by (\ref{condns}).
Now applying the mutation $\mu_2$ gives a new cluster variable 
\[
x_2''=\frac{P'_2(x_3,\ldots,x_{m+1})}{x_2}=\frac{P(x_3,\ldots,x_{m+1})}{x_2}, 
\]
which is defined to be $x_{m+2}$, and produces the new seed 
$$ 
\rho\circ \mu_2(\mathbf{x}',\mathbf{P}')  =
\Big((x_3,\ldots,x_{m},x_{m+1},x_{m+2}),(P''_3,\ldots,P''_{m},P_{1}'',P_2'')\Big), 
$$ 
where 
$$ 
P''_1=\mathcal{S}P'_{m}, \qquad  
P''_i=\mathcal{S}P_{i-1}'\quad  \mathrm{for} \quad 2\leq i \leq m. 
$$
Continuing to apply  consecutive mutations $\mu_3$, $\mu_4$, and so on, 
one can see that this will give precisely the iterates of (\ref{genericrecurrence}). 
Since these iterates are given by compositions of mutations they belong to the LP algebra 
$\mathcal{A}$ generated by the seed $(\mathbf{x},\mathbf{P})$, 
hence are Laurent polynomials in the initial cluster variables by Theorem \ref{LaurentProperty}.
\end{proof}

\subsection{Little Pi from a period $1$ seed}

The Little Pi recurrences (\ref{firstlittlepi}) are included in many examples of the form (\ref{genericrecurrence}) found in \cite{laurentphenomenonsequences} that can be shown to have the Laurent property by 
describing them in terms of 
successive  mutations of a period 1 seed, as in Proposition \ref{period1seed}. 
In order to apply this result they construct the ``intermediate polynomials", that is to say, the other 
exchange polynomials that appear in the period $1$ seed, such that the shifting conditions 
(\ref{condns}) hold. For Little Pi,  this construction is split in to four cases, which we list below.

For convenience, we slightly change the notation compared with the above discussion, where 
we followed \cite{laurentphenomenon} in labelling a cluster of size $m$ with indices from $1$ to $m$. To be consistent 
with \cite{laurentphenomenonsequences}, below we label the initial cluster variables $x_i$ 
and exchange polynomials $P_i$  with indices $0\leq i\leq m-1$.  
Note that the inclusion of the coefficient $a$ in (\ref{firstlittlepi}) means that the Laurent property takes the form 
$$x_n\in \Z [a, x_0^{\pm 1}, x_1^{\pm 1},\ldots, x_{2k+l-1}^{\pm 1}],$$
but in fact in the next section we will take $a\to 1$. 
(More details of the Laurent phenomenon over a ring of coefficients are provided in \cite{laurentphenomenon}.)   
Only the polynomials $P_j$ for $j\in J:=\{0,k,2k,l,k+l\}$ are given here. To find the intermediate polynomial 
$P_i$ for any $i$, take the largest $j\in J$ with $j\leq i$ and shift $P_j$ up by $i-j$, so that 
$P_i=\mathcal{S}^{i-j}P_j$. Note that in all cases we have $P_0=P$. 

\begin{itemize}
\item For $l>2k$:
\[
P_k=ax_0x_{2k}+ax_{2k}x_{l}+x_0x_{3k}x_{k+l}+a^2_{k+l}, 
\]
\[
P_{2k}=ax_0x_{3k}+ax_{l-k}x_{k+l}+x_0x_{l-2k}x_{k+l}+a^2x_0, 
\]
\[
P_{l}=ax_kx_{l-k}+ax_{l-k}x_{k+l}+x_0x_{l-2k}x_{k+l}+a^2x_0, 
\]
\[
P_{k+l}=ax_0+ax_{l}+x_kx_{l-k}. 
\]
\item For $l=2k$:
\[
P_k=ax_0x_{2k}+ax^2_{2k}+x_0x^2_{3k}+a^2x_{3k}, 
\]
\[
P_{2k}=ax^2_{k}+ax_kx_{3k}+x_0^2x_{3k}+a^2x_0, 
\]
\[
P_{3k}=ax_0+ax_{2k}+x^2_{k}.
\]
\item For $2k>l>k$:
\[
P_k=ax_0x_{2k}+ax_{2k}x_{l}+x_0x_{3k}x_{k+l}+a^2x_{k+l}, 
\]
\[
P_{k+l}=x_0x_{l-k}x_{k+l}+x_0x_{2l-k}x_{k+l}+x_{l-k}x_kx_{2l}+x_{l-k}x_{k+l}x_{2l}+ax_0x_{2l}, 
\]
\[
P_{2k}=ax_kx_{3k-l}+ax_kx_{3k}+x_0x_kx_{3k}+a^2x_{2k-l}, 
\]
\[
P_{k+l}=ax_0+ax_{l}+x_kx_{l-k}.
\]
\item For $k>l$:
\[
P_{l}=x_{2l}x_k+x_{2l}x_{k+l}+x_0x_{2k}+x_0x_{k+2l}, 
\]
\[
P_k=x_0x_{k+l}x_{2k-l}+x_0x_{k+l}x_{2k}+x_0x_{k-l}x_{2k}+x_{l}x_{k-l}x_{2k}+ax_{k-l}x_{k+l}, 
\]
\[
P_{k+l}=x_{l}x_k+x_{k+2l}x_{l}+x_0x_{k+2l}+x_kx_{2l}, 
\]
\[
P_{2k}=ax_{k-l}+ax_k+x_0x_{2k-l}. 
\]
\end{itemize}

\section{Linear relations with periodic coefficients for Little Pi}\label{littlepilin}

Henceforth we shall work with the Little Pi family of recurrences in the form 
\begin{equation}\label{10}
x_nx_{n+2k+l}=x_{n+2k}x_{n+l}+x_{n+k}+x_{n+k+l}, 
\end{equation}
which is obtained from (\ref{firstlittlepi}) after 
rescaling $x_n\rightarrow ax_n$.  
These generalize the family found by Heideman and Hogan \cite{heidemanhogan}, 
which is the case $l=1$. 
In order to find linear relations, we begin by showing that the $3 \times 3$ matrix 
\beq\label{psidef}
\Psi_n:=
\begin{bmatrix}
x_n & x_{n+2k} & x_{n+4k}  \\
x_{n+l} & x_{n+2k+l} & x_{n+4k+l}  \\
x_{n+2l} & x_{n+2k+2l} & x_{n+4k+2l} 
\end{bmatrix}
\eeq 
has a non-zero periodic determinant. 
For convenience we set 
\[
z_n:=x_n+x_{n+l}, 
\]
and 
note the following 
two identities which 
are a consequence of (\ref{10}): 
\begin{equation}\label{8}
z_nx_{n+2k+l}=x_{n+l}z_{n+2k}+z_{n+k}, 
\end{equation}
\begin{equation}\label{9}
z_nx_{n+2k}=x_nz_{n+2k}-z_{n+k}. 
\end{equation}

\begin{lemma}
The $3\times 3$ determinant
\begin{equation*}
\delta_n:=|\Psi_n|=\begin{vmatrix}
x_n & x_{n+2k} & x_{n+4k} \\
x_{n+l} & x_{n+2k+l} & x_{n+4k+l} \\
x_{n+2l} & x_{n+2k+2l} & x_{n+4k+2l}
\end{vmatrix}
\end{equation*}
has period $k$.
\end{lemma}
\begin{proof} 
First observe that 
(\ref{10}) 
can be rewritten as  
\[
\begin{vmatrix}
x_n & x_{n+2k} \\
x_{n+l} & x_{n+2k+l}
\end{vmatrix}
=z_{n+k}, 
\] 
so using Dodgson condensation, as given by (\ref{dc}) with $N=3$, 
we may write
\begin{equation}\label{3}
x_{n+2k+l}\delta_n=
\begin{vmatrix}
z_{n+k} & z_{n+3k} \\
z_{n+k+l} & z_{n+3k+l}
\end{vmatrix}=:\delta'_{n+k}.
\end{equation}
Upon scaling the first column by $x_{n+3k+l}$, we see that 
the $2\times 2$ determinant $\delta'_{n+k}$ in (\ref{3}) satisfies 
\begin{equation*}
x_{n+3k+l}\delta'_{n+k}=
\begin{vmatrix}
z_{n+k}x_{n+3k+l} & z_{n+3k} \\
z_{n+k+l}x_{n+3k+l} & z_{n+3k+l}
\end{vmatrix}.
\end{equation*}
Then we can use (\ref{8}) and (\ref{9}) on the left column to obtain 
\begin{equation}\label{14}
x_{n+3k+l}\delta'_{n+k}=
\begin{vmatrix}
x_{n+k+l}z_{n+3k}+z_{n+2k} & z_{n+3k} \\
x_{n+k+l}z_{n+3k+l}-z_{n+2k+l} & z_{n+3k+l}
\end{vmatrix}
=
\begin{vmatrix}
z_{n+2k} & z_{n+3k} \\
-z_{n+2k+l} & z_{n+3k+l}
\end{vmatrix}, 
\end{equation}
and 
by the same token, but instead manipulating the right column in (\ref{3}), we have
\begin{equation} \label{dprime} 
x_{n+k+l}\delta'_{n+k}=
\begin{vmatrix}
z_{n+k} & -z_{n+2k} \\
z_{n+k+l} & z_{n+2k+l}
\end{vmatrix}. 
\end{equation} 
Shifting up $n \to n+k$ and comparing with (\ref{14}) we arrive at  
\[
\frac{\delta'_{n+k}}{x_{n+2k+l}}=\frac{\delta'_{n+2k}}{x_{n+3k+l}}, 
\]
so these ratios are periodic with period $k$, which is the required result.
\end{proof}

\begin{lemma} \label{nonzerolemma}
For each $n$ the determinant $\delta_n=|\Psi_n|$ is non-zero, 
considered as an element of $\mathbb{Q}(x_0,x_1,\ldots,x_{2k+l-1})$, the ambient field of fractions in the initial data 
for (\ref{10}). 
\end{lemma} 
\begin{proof} 
Without assuming the Laurent property, a priori the iterates of (\ref{10}) are rational functions of the 
initial data with rational numbers as coefficients, and the same is true for the determinant $\delta_n$. Let us consider the case of 
substituting real positive initial values $x_n>0$ for $n=0,\ldots,2k+l-1$. It follows by induction that $x_n>0$ for all $n\in\mathbb{Z}$, 
hence also  $z_n>0$ for all $n$. If $\delta_n$ vanishes for some $n$ then $\delta_{n+k}'$ vanishes, by (\ref{3}), but then 
$$ 
z_{n+k}z_{n+2k+l}+z_{n+2k}z_{n+k+l}=0
$$ 
by (\ref{dprime}), which is a contradiction. 
Hence $\delta_n$ is a non-zero rational function. 
\end{proof}


We now consider the corresponding $4\times 4$ matrix
\[
\hat{\Psi}_n:=
\begin{bmatrix}
x_n & x_{n+2k} & x_{n+4k} & x_{n+6k} \\
x_{n+l} & x_{n+2k+l} & x_{n+4k+l} & x_{n+6k+l} \\
x_{n+2l} & x_{n+2k+2l} & x_{n+4k+2l} & x_{n+6k+2l} \\
x_{n+3l} & x_{n+2k+3l} & x_{n+4k+3l} & x_{n+6k+3l} 
\end{bmatrix} ,
\]
and use Dodgson condensation once more, with $N=4$ in (\ref{dc}), to calculate  
\[
|\hat{\Psi}_n|=\frac{\delta_{n+k}\delta_{n+2k+l}-\delta_{n+k+l}\delta_{n+2k}}{z_{n+3k+l}}, 
\]
and then by periodicity of $\delta_n$ we have the 
\begin{corollary}
The $4\times 4$ determinant $|\hat{\Psi}_n|$ is identically zero. 
\end{corollary}

Given that $|\hat{\Psi}_n|=0$, we obtain 
linear relations with periodic coefficients 
by considering 
the right and left kernels (i.e.\ the kernel of $\hat{\Psi}_n$ and that of its transpose). 

\begin{remark}
The kernel of $\hat{\Psi}_n$ is one-dimensional, since if it were of dimension greater than one then  $\Psi_n$ 
we would have a non-trivial kernel, contradicting Lemma \ref{nonzerolemma}.
\end{remark}

\begin{proposition}\label{periodickernel}
The iterates of (\ref{10}) satisfy the linear relations  
\begin{equation}\label{rel1} 
x_{n+6k}+K^{(3)}_{n}x_{n+4k}+K^{(2)}_nx_{n+2k}-x_n=0, 
\end{equation}
\begin{equation}\label{rel2} 
x_{n+3l}+\gamma_nx_{n+2l}+\beta_nx_{n+l}+\alpha_nx_n=0, 
\end{equation}
with periodic coefficients: $K^{(2)}_n$ and $K^{(3)}_n$ have period $l$, $\alpha_n$ has period $k$, and $\beta_n$ and $\gamma_n$ have period $2k$.
\end{proposition}
\begin{proof}
Let $(K^{(1)}_n,K^{(2)}_n,K^{(3)}_n,1)^T$ be in the kernel of $\hat{\Psi}_n$. (We are justified in scaling the last entry to $1$ due to Lemma \ref{nonzerolemma}.)
From the first three rows of 
\begin{equation}
\label{6}\hat{\Psi}_n(K^{(1)}_n,K^{(2)}_n,K^{(3)}_n,1)^T=0
\end{equation}  
we get the matrix equation
\begin{equation}\label{4}
\begin{bmatrix}
x_n & x_{n+2k} & x_{n+4k} \\
x_{n+l} & x_{n+2k+l} & x_{n+4k+l} \\
x_{n+2l} & x_{n+2k+2l} & x_{n+4k+2l} 
\end{bmatrix}
\begin{bmatrix}
K^{(1)}_n \\
K^{(2)}_n \\
K^{(3)}_n
\end{bmatrix}
=-\begin{bmatrix}
x_{n+6k} \\
x_{n+6k+l} \\
x_{n+6k+2l}
\end{bmatrix}, 
\end{equation}
and by Cramer's rule 
\[
K^{(1)}_n=\frac{-\delta_{n+2k}}{\delta_n}=-1
\]
The last $3$ rows of (\ref{6}) give
\begin{equation}\label{5}
\begin{bmatrix}

x_{n+l} & x_{n+2k+l} & x_{n+4k+l} \\
x_{n+2l} & x_{n+2k+2l} & x_{n+4k+2l}\\
 x_{n+3l} & x_{n+2k+3l} & x_{n+4k+3l}
\end{bmatrix}
\begin{bmatrix}
K^{(1)}_n \\
K^{(2)}_n \\
K^{(3)}_n
\end{bmatrix}
=-\begin{bmatrix}
x_{n+6k+l} \\
x_{n+6k+2l} \\
x_{n+6k+3l}
\end{bmatrix}.
\end{equation}
The equations (\ref{4}) and (\ref{5}) imply that $K^{(2)}_n$ and $K^{(3)}_n$ both have period $l$.  
Now set 
\begin{equation}\label{7}
\hat{\Psi}_n^T(\alpha_n,\beta_n,\gamma_n,1)^T=0, 
\end{equation}
and analogous 
arguments to the preceding ones give 
\beq\label{aldef} 
\alpha_n=-\frac{\delta_{n+l}}{\delta_n}
\eeq 
and the result that $\alpha_n$ is $k$-periodic, and $\beta_n$ and $\gamma_n$ are $2k$-periodic. 
\end{proof}

We can derive further relations between the coefficients in (\ref{rel1}) and (\ref{rel2}) by using  (\ref{4}) and (\ref{5}), 
as well as  the corresponding equations for the 
left kernel of $\hat{\Psi}_n$.
\begin{lemma}\label{coefficientrelation}
The periodic coefficients in (\ref{rel1}) are related to one another by $K^{(2)}_{n+k}=-K^{(3)}_n$. 
\end{lemma}
\begin{proof}
From the first $2$ rows of (\ref{6})
we have 
\[
\begin{bmatrix}
x_{n+2k} & x_{n+4k} \\
x_{n+2k+l} & x_{n+4k+l}
\end{bmatrix}
\begin{bmatrix}
K^{(2)}_n \\
K^{(3)}_n
\end{bmatrix}=
\begin{bmatrix}
x_n & x_{n+6k} \\
x_{n+l} & x_{n+6k+l}
\end{bmatrix}
\begin{bmatrix}
1 \\
-1
\end{bmatrix}
\]
Then 
solving for $K^{(2)}_n$ and $K^{(3)}_n$ yields   
\[
K^{(2)}_n=\frac{x_nx_{n+4k+l}-x_{n+l}x_{n+4k}+z_{n+5k}}{z_{n+3k}}, 
\]
\[
K^{(3)}_n=\frac{x_{n+2k+l}x_{n+6k}-x_{n+2k}x_{n+6k+l}-z_{n+k}}{z_{n+3k}}, 
\]
from which we get
\[
z_{n+5k}K^{(2)}_{n+2k}+z_{n+3k}K^{(3)}_n=z_{n+7k}-z_{n+k}. 
\]
The sequence of $z_j$ satisfy the same matrix equation (\ref{4}) as $x_j$, obtained 
by replacing each $x_j\to z_j$ in (\ref{4}), due to the $l$-periodicity of $K^{(2)}$ and $K^{(3)}$, so
\[
z_{n+5k}K^{(2)}_{n+2k}+z_{n+3k}K^{(3)}_n=-K^{(2)}_{n+k}z_{n+3k}--K^{(3)}_{n+k}z_{n+5k}
\]
Assuming that $K_{n+k}^{(2)}+K^{(3)}_n\neq 0 $ implies 
\[
\frac{K_{n+2k}^{(2)}+K_{n+k}^{(3)}}{K_{n+k}^{(2)}+K_{n}^{(3)}}=-\frac{z_{n+3k}}{z_{n+5k}}, 
\]
and 
the left-hand side above is periodic with period $l$ so the right-hand side should be too, i.e.
\[
\frac{z_{n+3k}}{z_{n+5k}}=\frac{z_{n+3k+l}}{z_{n+5k+l}} \quad \iff \quad
\begin{vmatrix}
z_{n+3k} & z_{n+3k+l} \\
z_{n+5k} & z_{n+5k+l}
\end{vmatrix}=0, 
\]
and this determinant is $\delta'_{n+3k}$ from (\ref{3}), but by the proof of Lemma \ref{nonzerolemma} this cannot be 
identically zero, which gives a contradiction. Hence $K_{n+k}^{(2)}+K^{(3)}_n= 0$ as required.
\end{proof}
\begin{corollary}\label{periodlrelation} 
The linear relation (\ref{rel1}) with $l$-periodic coefficients has the form    
\[
x_{n+6k}-K_{n+k}x_{n+4k}+K_nx_{n+2k}-x_n=0
\]
with $K_n:=K^{(2)}_n$.
\end{corollary}

\begin{remark} The latter results were previously obtained via a different method, using the travelling wave reduction of (\ref{firstlittlepi}), 
in \cite{kkmt} (see 
Corollary 3.2 and Proposition 3.3 therein).  
\end{remark} 

We close this section by proving some conjectures made for $l=1$ in \cite{wardthesis}, and extending them to arbitrary $l$. 
\begin{proposition}\label{abcid} 
The periodic coefficients of the linear relation (\ref{rel2}) 
satisfy the following set of identities:
\begin{equation}\label{greekrelation1}
\alpha_n=\beta_n+\gamma_{n+k}-1, 
\end{equation} 
\begin{equation}\label{greekrelation2}
\alpha_{n+l}( {\gamma_n+\gamma_{n+k}} ) =
{\beta_{n+l}+\beta_{n+k+l}} , 
\end{equation}
\beq\label{productlemma}
\prod_{i=0}^{k-1}\alpha_{n+i} 
=(-1)^{k}. 
\eeq
\end{proposition}
\begin{proof}
From the left kernel analogue of (\ref{5}) we have 
\[
\begin{bmatrix}
x_{n+l} & x_{n+2l} \\
x_{n+2k+l} & x_{n+2k+2l} 
\end{bmatrix}
\begin{bmatrix}
\beta_n \\
\gamma_n
\end{bmatrix}
=-
\begin{bmatrix}
x_n & x_{n+3l} \\
x_{n+2k} & x_{n+2k+3l}
\end{bmatrix}
\begin{bmatrix}
\alpha_n \\
1
\end{bmatrix}, 
\]
so we can express $\beta_n$ and $\gamma_n$ as
\[
\beta_n=\frac{\alpha_n(x_{n+2k}x_{n+2l}-x_nx_{n+2k+2l})+z_{n+k+2l}}{z_{n+k+l}}, 
\]
\[
\gamma_n=\frac{(x_{n+2k+l}x_{n+3l}-x_{n+l}x_{n+2k+3l})+\alpha_nz_{n+k}}{z_{n+k+l}}. 
\]
Upon shifting $\beta_n \to \beta_{n+k}$ 
we can equate the bracketed terms above as 
\begin{equation}\label{20}
\alpha_{n+l}\gamma_nz_{n+k+l}-\alpha_n\alpha_{n+l}z_{n+k}=\beta_{n+l}z_{n+k+2l}-z_{n+k+3l}. 
\end{equation}
Now if we write the $z_j$ in terms of the $x_i$ and replace the $x_{n+k+4l}$ that appears as
\[
x_{n+k+4l}=-\alpha_{n+l}x_{n+k+l}-\beta_{n+k+l}x_{n+k+2l}-\gamma_{n+k+l}x_{n+k+3l}, 
\]
then (\ref{20}) becomes
\begin{multline}\label{21}
-\alpha_n\alpha_{n+l}x_{n+k}+(\alpha_{n+l}\gamma_n-\alpha_n\alpha_{n+l}-\alpha_{n+l})x_{n+k+l}\\+(\alpha_{n+l}\gamma_n-\beta_{n+l}-\beta_{n+k+l})x_{n+k+2l}+(1-\beta_{n+l}-\gamma_{n+k+l})x_{n+k+3l}=0. 
\end{multline}
Since the kernel of $\hat{\Psi}_n$ is one-dimensional we can scale and equate coefficients in (\ref{21}) and an appropriate shift of (\ref{7}) to get three equations, namely 
\[
\alpha_n=\beta_n+\gamma_{n+k}-1, \,\,
\gamma_{n+k}\alpha_{n+l}=\beta_{n+l}+\beta_{n+k+l}-\alpha_{n+l}\gamma_n,\,\,  
\alpha_{n+l}=\beta_{n+l}+\gamma_{n+k+l}-1. 
\]
where the third of these is simply a shift of the first, and these rearrange to give 
(\ref{greekrelation1}) and (\ref{greekrelation2}). The identity (\ref{productlemma}) 
follows from (\ref{aldef}) and the fact that $\delta_n$ has period $k$. 
\end{proof}

\section{Linear relations with constant  coefficients} \label{littlepiconst}

In this section we derive the linearization of the Little Pi family (\ref{10}), in the form 
of linear relations with constant coefficients, which were not previously considered in \cite{kkmt}. 
The key is to use monodromy arguments, similar to those employed in \cite{fordyhone} in the case 
of the cluster algebra recurrences (\ref{affa}) and (\ref{2fri}). 

We start by defining the sequences of matrices 
\beq\label{lltdef} 
L_n:=
\begin{bmatrix}
0 & 0 & 1 \\
1 & 0 & -K_n \\
0 & 1 & K_{n+k}
\end{bmatrix}, 
\qquad
\tilde{L}_n:=
\begin{bmatrix}
0 & 1 & 0 \\
0 & 0 & 1 \\
-\alpha_n & -\beta_n & -\gamma_n
\end{bmatrix}, 
\eeq 
which vary with overall periods $l$ and $2k$, respectively, and 
allow the linear relations (\ref{rel1}) and (\ref{rel2}) to be rewritten in 
matrix form as 
\[
\Psi_nL_n=\Psi_{n+2k},  \qquad \tilde{L}_n\Psi_n=\Psi_{n+l}, 
\]
where as before $\Psi_n$ is given by (\ref{psidef}).
The point of this is that if we define
the pair of monodromy matrices 
\begin{equation}\label{M_ndef}
M_n:=L_nL_{n+2k}L_{n+4k}\cdots L_{n+2k(l-1)},  \qquad \tilde{M}_n:=\tilde{L}_{n+(2k-1)l} \cdots \tilde{L}_{n+2l}\tilde{L}_{n+l}\tilde{L}_n
\end{equation}
then right multiplication by $M_n$ will shift $\Psi_n$ by $2k$ upwards $l$ times, that is 
\begin{equation} \label{mpsi}
\Psi_nM_n=\Psi_{n+2kl}, 
\end{equation}
and left multiplication by $\tilde{M}_n$  will shift $\Psi_n$ by $l$ upwards $2k$ times, so that 
\begin{equation} \label{tmpsi}
\tilde{M}_n\Psi_n=\Psi_{n+2kl}. 
\end{equation}

\begin{remark}
If $d:=\gcd(k,l)>1$ then the recurrence (\ref{10}) splits into $d$ copies of itself, so without loss of generality we can take $d=1$. 
Then with $d=1$, if $l$ is odd then $\gcd(2k,l)=1$ and $\lcm(2k,l)=2kl$, while if $l$ is even then $\gcd(2k,l)=2$ and $\lcm(2k,l)=kl$, 
and we need to deal with these two different cases separately.
\end{remark}

\subsection{The case $\gcd(2k,l)=1$}
Here $\lcm(2k,l)=2kl$, so $l$ is odd, and with the monodromy matrices 
$M_n$ and $\tilde{M}_n$ defined as in (\ref{M_ndef}) above, we see that due to the $l$-periodicity of $L_n$ and the cyclic property of the trace, 
the quantity $\mathcal{K}:= \mathrm{tr}(M_n)$ has period $l$, and similarly $\mathrm{tr}(\tilde{M}_n)$ has period $2k$. Now from 
(\ref{mpsi}) and (\ref{tmpsi}) we have 
$$ 
\mathcal{K}=\mathrm{tr}(M_n)=\mathrm{tr}(\Psi_n^{-1}\Psi_{n+2kl}) = \mathrm{tr}(\Psi_{n+2kl}\Psi_n^{-1})=\mathrm{tr}(\tilde{M}_n), 
$$ 
so $\mathcal{K}$ has period $\gcd(2k,l)=1$, hence is a first integral for (\ref{10}), independent of $n$. 
The same argument applies  to the quantity $\mathcal{\tilde{K}}:=\tr(M_n^{-1})=\tr(\tilde{M}_n^{-1})$, 
which we will now show is equal to $\mathcal{K}$. 

\begin{proposition} \label{traceequality} 
The trace of the monodromy matrix $M_n$ satisfies 
$\mathcal{K}=\mathrm{tr}(M_n) =\tr(M_n^{-1})$. 
\end{proposition} 

\begin{proof}
The result holds in the case $l=1$,  since we have $\mathcal{K}=\mathrm{tr}(M_n)=\mathrm{tr}(L_n)=K_{n+k}$, $\mathcal{\tilde{K}}=\mathrm{tr}(L_n^{-1})=K_n$, and 
$K_n$ has period $1$, hence   $\mathcal{K}=\mathcal{\tilde{K}}$.  (Another proof for  $l=1$ is given in \cite{heidemanhoganhone}.) Rewriting $\mathcal{K}$ as $\mathrm{tr}(\tilde{M}_n)$, 
and similarly for $\mathcal{\tilde{K}}=\mathrm{tr}(\tilde{M}_n^{-1})$, 
and (setting $n=0$ without loss of generality) 
this implies an algebraic relation between the $5k$  quantities 
$\al_0,\ldots, \al_{k-1}, \beta_0,\ldots , \beta_{2k-1}, \gamma_0,\ldots, \gamma_{2k-1}$, namely that 
\beq\label{traceid} 
\tr (\tilde{L}_{2k-1} \tilde{L}_{2k-2} \cdots \tilde{L}_{0})= \tr( \tilde{L}_{0}^{-1} \tilde{L}_{1}^{-1} 
\cdots \tilde{L}_{2k-1}^{-1} )
\eeq    
must hold as a consequence of the relations in Proposition \ref{abcid} for $l=1$. 
Due to periodicity there are $2k$ independent relations of the form (\ref{greekrelation1}), as well 
as $k$ independent relations of the form   (\ref{greekrelation2}), together with  (\ref{productlemma}), 
but in fact there are only $3k$ independent relations in total, so these equations define an affine variety of 
dimension $2k$.  
Now for odd $l>1$ note that for $\mathcal{K}=\mathcal{\tilde{K}}$ to holds if and only if  
\beq\label{trace2}
\tr (\tilde{L}_{\si (2k-1)} \tilde{L}_{\si (2k-2)} \cdots \tilde{L}_{\si(0)})= \tr (\tilde{L}_{\si(0)}^{-1} \tilde{L}_{\si(1)}^{-1} 
\cdots \tilde{L}_{\si(2k-1)}^{-1}),
\eeq 
where $\si$ is the permutation of the indices $0,1,\ldots,2k-1$ defined by 
$$ 
\si(i) = il \bmod 2k, 
$$ 
which satisfies the properties 
$$ 
\si(i+k)-\si(i) \equiv k \, (\bmod \,2k), \qquad 
\si(i+1)-\si(i) \equiv l \, (\bmod \,2k). 
$$ 
With all indices read mod $2k$ (or mod $k$ in the case of $\al_j$), 
it follows from these properties that $\si$ acts  by permuting the coordinates $\al_j,\beta_j,\gamma_j$ in the identities for 
$l=1$ in Proposition \ref{abcid}, 
so that the identities for each odd $l$ are just 
$$ 
\al_{\si(n)} = \beta_{\si(n)}+    \gamma_{\si(n+k)} -1, 
\qquad 
\al_{\si(n+1)} (\gamma_{\si(n)} +\gamma_{\si(n+k)}) 
= \beta_{\si(n+1)} +\gamma_{\si(n+k+1)}, 
$$ 
and similarly for (\ref{productlemma}). In other words, the identities for odd $l>1$ are just permutations of those 
for $l=1$, so the algebraic relation (\ref{trace2}) holds as an immediate consequence of the relation 
(\ref{traceid}) when $l=1$.   
\end{proof}

\begin{remark}
There is an implicit assumption in the above proof, namely that 
when $l=1$ the map from the initial values $x_0,x_1,\ldots, x_{2k}$ for (\ref{10}) 
to the variety defined by  the relations in Proposition \ref{abcid} is surjective, 
which ensures that the identity $\mathrm{tr}(M_n) =\tr(M_n^{-1})$ must be an algebraic 
consequence of these relations. In particular, it is enough to check 
that for $l=1$ there is collection of $2k$ independent $2k$-periodic functions of the 
initial data (e.g.\ either of the sets
$\beta_0,\ldots,\beta_{2k-1}$ or  
$\gamma_0,\ldots,\gamma_{2k-1}$ should be functionally independent). 
While this is a straightforward but laborious task for any given $k$, we do not know of a 
simple verification that is 
valid for all $k$. However, for any $l$, a direct algebraic proof of Proposition  \ref{traceequality}  
is provided by the argument used to prove Theorem 4.5 in \cite{honeward}, which we will revisit 
in the proof of Theorem \ref{extlin} below. 
\end{remark} 

\begin{theorem}\label{linearrelationtheorem}
If $\gcd(2k,l)=1$ then the iterates of (\ref{10}) satisfy
the constant coefficient linear relation 
\[
x_{n+6kl}-\mathcal{K}x_{n+4kl}+\mathcal{K}x_{n+2kl}-x_n=0
\] 
where $\mathcal{K}=\mathrm{tr}(M_n)$ is a first integral.
\end{theorem} 

\begin{proof} 
Note that $|L_n|=1$, hence $|M_n|=1$, so  
by the Cayley-Hamilton theorem applied to $M_n$ we have
\begin{equation}\label{13}
M_n^3-\tr(M_n)M_n^2+cM_n-I=0
\end{equation} 
for some $c$. Premultiplying by $M_n^{-3}$ in (\ref{13})
gives 
\[
M_n^{-3}-cM_n^{-2}+\tr(M_n)M_n^{-1}-I=0, 
\] 
so from the  Cayley-Hamilton theorem for $M_n^{-1}$
we see that $c=\tr(M^{-1}_n)=\tr(M_n)=\mathcal{K}$. Multiplying (\ref{13}) by $\Psi_n$ from the left yields
\[
\Psi_{n+6kl}-\mathcal{K}\Psi_{n+4kl}+\mathcal{K}\Psi_{n+2kl}-\Psi_n=0
\] 
and the top leftmost entry of this matrix equation gives the required linear relation for $x_n$.
\end{proof}

\subsection{The case $\gcd(2k,l)=2$}

In this case $\lcm(2k,l)=kl$, with $l$ even and $k$ odd. 
With the same definition (\ref{M_ndef}) for $M_n$, each matrix in the product 
appears twice in the same 
order relative to its neighbours, so  $M_n$ is a perfect square, 
and we can define the square root $M_n^*=M_n^{1/2}$ by the same product with half as many factors, and similarly 
for $\tilde{M}_n^*=\tilde{M}_n^{1/2}$. 
The total shift for $\Psi_n$ is now $kl$ instead of $2kl$, so we have   
\[
M_n^*:=L_nL_{n+2k}\cdots L_{n+k(l-4)}L_{n+k(l-2)}, \qquad \tilde{M}_n^*:=\tilde{L}_{n+(k-1)l}\tilde{L}_{n+(k-2)l}
\cdots \tilde{L}_{n+l}\tilde{L}_n, 
\]
with
\begin{equation}\label{psis}
\Psi_{n+kl}=\Psi_nM_n^*=\tilde{M}_n^*\Psi_n
\end{equation}
Again $\tr(M_n^*)$ has period $2k$ and $\tr(\tilde{M}_n^*)$ has period $l$, but now this implies that  $\mathcal{K}_n:=\tr(M_n^*)=\tr(\tilde{M}^*_n)$ has period $\gcd(2k,l)=2$, 
and similarly for $\tilde{\mathcal{K}}_n:=\tr((M_n^*)^{-1})$. The analogue of Proposition \ref{traceequality} requires more work in this case. We begin with

\begin{proposition} \label{traceequality2} 
The trace of 
$M_n^*$ satisfies 
$\mathcal{K}_n=\mathrm{tr}(M_n^*) =\tr((M_{n+1}^*)^{-1})$. 
\end{proposition}
\begin{proof}	
When $k=1$  in terms of $\tilde{M}^*_n$ we have 
\[
\mathcal{K}_n=\tr(\tilde{L}_n)=-\gamma_{n}, \qquad 
\mathcal{\tilde{K}}_n=\tr(\tilde{L}_n^{-1})=-\frac{\beta_n}{\alpha_n} .
\]
Now due to (\ref{productlemma}) we have $\alpha_n=-1$ and using (\ref{greekrelation1}) we get $\beta_n=-\gamma_{n+1}$, hence $\mathcal{K}_n=\mathcal{\tilde{K}}_{n+1}$, 
and the result holds in this case. This implies an algebraic identity between the entries of the 
monodromy matrix $M^*_n$ and its shift, namely that  
$$ 
\tr (L_0L_2\cdots L_{l-2}) = \tr ( L_{l-1}^{-1}  L_{l-3}^{-1}\cdots L_1^{-1})
$$  
(where we set $n=0$ without loss of generality), which is just a tautology in terms of the 
$l$ quantities $K_0,K_1,\ldots,K_{l-1}$ that appear in the entries. 
Similarly to the argument in the proof of Proposition \ref{traceequality}, we have that 
for $k>1$ the required identity of traces for $M^*_n$ and $(M^*_{n+1})^{-1}$ 
just corresponds to a permutation $\si$ of the indices 
of the quantities $K_j$, given by $\si(j)=2jk\bmod l$, so 
the relation   $\mathcal{K}_n=\mathcal{\tilde{K}}_{n+1}$ holds for all $k$. 
\end{proof}

\begin{proposition}
The iterates of (\ref{10}) satisfy
a linear relation with period 2 coefficients, given by 
\begin{equation}\label{18}
x_{n+3kl}-\mathcal{K}_nx_{n+2kl}+\mathcal{K}_{n+1}x_{n+kl}-x_n=0
\end{equation}
\end{proposition}
\begin{proof}
This follows by the Cayley-Hamilton theorem, as in the proof of  Theorem \ref{linearrelationtheorem}, but with different traces appearing.
\end{proof}

\begin{theorem} When $\gcd(2k,l)=2$, the iterates of (\ref{10}) satisfy 
the constant coefficient relation
\[
x_{n+6kl}-\mathcal{A}x_{n+5kl} +\mathcal{B}x_{n+4kl}-\mathcal{C}x_{n+3kl}+\mathcal{B}x_{n+2kl}-\mathcal{A}x_{n+kl}+x_n=0
\]
where
\[
\mathcal{A}=\mathcal{K}_n+\mathcal{K}_{n+1}, \qquad \mathcal{B}=\mathcal{K}_n\mathcal{K}_{n+1}+\mathcal{K}_n+\mathcal{K}_{n+1}, \qquad \mathcal{C}=\mathcal{K}^2_n+\mathcal{K}^2_{n+1}+2.
\]
\end{theorem}
\begin{proof}
Let $\mathcal{S}$ be the shift operatorm such that 
$\mathcal{S}(f_n) =f_{n+1}$ for any function of $n$. Then applying
the operator 
\[
\mathcal{S}^{3kl}-\mathcal{K}_{n+1}\mathcal{S}^{2kl}+\mathcal{K}_{n}\mathcal{S}^{kl}-1
\]
to equation (\ref{18}) gives the required result.
\end{proof}

\subsection{Superintegrability of Little Pi}

The birational map 
$$ 
\varphi: \quad (x_0,\ldots, x_{2k+l-2},x_{2k+l-1}) \mapsto 
(x_1,\ldots, x_{2k+l-1},x_{2k+l})
$$ 
defined by the Little Pi recurrence (\ref{10}) is measure-preserving, in the sense that 
$$ 
\varphi^* \Omega = (-1)^l \Omega, 
$$ 
where $\Omega$ is the volume form 
$$ 
\Omega = \frac{\rd x_0\wedge \rd x_1 \wedge \cdots \wedge \rd x_{2k+l-1}}{x_0x_1\cdots x_{2k+l-1}}. 
$$ 
We can use this to show that the map $\varphi$ is maximally superintegrable, in the sense that 
it admits an (anti-) invariant Poisson structure, and the number of independent first integrals is one less than 
the dimension of the phase space. 

In the case $\gcd(2k,l)=1$, it appears that the $l$-periodic quantities $K_0,\ldots, K_{l-1}$ are independent 
of one another, hence any cyclically symmetric functions of these quantities are first integrals: so this provides 
$l$ independent first integrals for (\ref{10}). Similarly, subject to the 
relations in Proposition \ref{abcid} one can take $2k$ independent $2k$-periodic quantities, and cyclically 
symmetric functions of these provide 
$2k$ independent first integrals. However, in total this should give exactly $2k+l-1$ independent 
first integrals $I_1,I_2,\ldots,I_{2k+l-1}$, since the identity    
$\tr (M_n) = \tr (\tilde{M}_n)$ gives a relation between these two sets of cyclically symmetric functions. 
Then by a result from \cite{byrnes}, taking all but one of these first integrals together with 
the covolume form 
$$ 
V = x_0\cdots x_{2k+l-1} \,\frac{\partial}{\partial x_0} \wedge \cdots \wedge  \frac{\partial}{\partial x_{2k+l-1}} 
$$ 
(i.e.\ the $(2k+l)$-multivector field that contracts with $\Omega$ to give 1) yields a Poisson bracket 
defined by 
$$ 
\{\,f,g\,\}=V(\rd f,\rd g,\rd I_1,\rd I_2,\ldots, \rd  I_{2k+l-2}), 
$$
and this bracket is invariant/anti-invariant under the action of $\varphi$, according to the 
parity of $l$, that is 
$$ 
\varphi^*\{\,f,g\,\} = (-1)^l\,\{\,\varphi^*f,\varphi^*g\,\}
$$ 
for any pair of functions $f,g$ on the $(2k+l)$-dimensional phase space. 
By construction, the first integrals $I_1,I_2,\ldots,I_{2k+l-2}$ are Casimirs 
for this bracket, but the additional first integral $I_{2k+l-1}$ is not, so 
it defines a non-trivial Hamiltonian vector field. 

As an example, we take the simplest case $k=l=1$, when $\varphi$ is defined by 
\beq\label{simple} 
x_{n+3}x_n = x_{n+2}x_{n+1} + x_{n+2}+x_{n+1}. 
\eeq 
 The quantity 
$K_0=\mathcal{K}$ is a first integral, which can be written as 
$$ 
\mathcal{K}=\frac{x_{6}-x_0}{x_4-x_2}, 
$$ 
by Theorem \ref{linearrelationtheorem}, and 
then rewritten as a function of the initial values $x_0,x_1,x_2$ by using  (\ref{simple}). In fact, by Theorem 1.2 
in \cite{heidemanhoganhone}, the explicit expression is 
$$ 
\mathcal{K}=P^{(0)}+P^{(1)}+P^{(2)}, 
$$ where 
$$
P^{(0)}=1+\frac{x_0}{x_2}+\frac{x_2}{x_0}, 
\quad P^{(1)}=\left(1+\frac{x_2}{x_0}\right)\, \frac{x_0+x_1}{x_1x_2} + 
\left(1+\frac{x_0}{x_2}\right)\, \frac{x_1+x_2}{x_0x_1}, 
$$
and 
$$ 
P^{(2)}= \frac{1}{x_1x_2} +\frac{1}{x_0x_1} +\frac{1}{x_0x_2}.
$$  
Also, by Proposition \ref{abcid} we have  
$$\al_0=-1, \qquad \gamma_0=-\beta_1,\qquad \gamma_1=-\beta_0,$$ 
and then from (\ref{rel2}) we can find two independent 2-periodic 
quantities by solving for $\beta_0,\beta_1$ in terms of the $x_j$ from the pair 
of linear equations 
$$ 
\begin{array}{rcl} 
x_3-\beta_1\,x_2+\beta_0\,x_1-x_0 & = & 0, \\ 
x_4-\beta_0\,x_3+\beta_1\,x_2-x_1 & = & 0. 
\end{array} 
$$
Then the two symmetric functions 
$$I_1=\beta_0\beta_1, \qquad I_2=\beta_0+\beta_1$$ 
provide two independent first integrals, but they are related to $\mathcal{K}$ 
by 
$$ 
\mathcal{K}=\tr(M_0) = \tr(\tilde{M}_0) = \tr (\tilde{L}_1\tilde{L}_0)=\beta_0\beta_1-\beta_0-\beta_1=I_1-I_2. 
$$ 
Finally, contracting the covolume form 
$$ 
V=x_0x_1x_2\, \frac{\partial}{\partial x_0} \wedge  \frac{\partial}{\partial x_{1}} \wedge  \frac{\partial}{\partial x_{2}}
$$ 
with the one-form $\rd I_1$ gives the Poisson bracket 
$$ 
\{\,f,g\,\}=V(\rd f,\rd g,\rd I_1),  
$$ 
which is anti-invariant under the action of the map $\varphi$ defined by (\ref{simple}), so 
it is invariant under the doubled map $\varphi^2$, and this is a superintegrable map in three 
dimensions. The flow of the Hamiltonian vector field 
$$ 
\frac{\rd }{\rd t} = \{ \, \cdot , I_2 \, \} 
$$ commutes with the map, and its level sets are curves defined by $I_1=\,$const, $I_2=\,$const.

\section{Linearization and reductions of a   6-point lattice equation}\label{latticesection}

In this section we consider the new 6-point lattice equation (\ref{lattice}), which can be rewritten as 
an equality of two $2\times 2$ determinants, in the form 
$$ 
\left| \begin{array}{cc} u_{s,t+1}  & u_{s,t+2} \\ 
u_{s+1,t+1}& u_{s+1,t+2} +a 
\end{array} \right|= 
\left| \begin{array}{ccc} u_{s,t}+a &u_{s,t+1} \\ 
u_{s+1,t}&u_{s+1,t+1} 
\end{array} \right|, 
$$
or in the form of a conservation law, as  
\beq\label{conslaw}
 \Delta_s \, au_{s,t+1}
=\Delta_t \left| \begin{array}{ccc} u_{s,t}&u_{s,t+1} \\ 
u_{s+1,t}&u_{s+1,t+1} 
\end{array} \right|.
\eeq 
By imposing the constraint 
$$
u_{s,t}=u_{s+k,t-l}   
$$ 
for integers $k,l$, one obtains the $(l,-k)$ travelling wave reduction 
\begin{equation}\label{wavereduction}
u_{s,t}=x_n \qquad n=ls+kt, 
\end{equation}
which produces the family of recurrences (\ref{1drecurrence}). Upon making use 
of the conservation law (\ref{conslaw}), we can write 
the reduction as 
$$ 
(\mathcal{S}^k-1) \left| \begin{array}{ccc} x_n & x_{n+k} \\ 
x_{n+l}&x_{n+k+l} 
\end{array} \right|=(\mathcal{S}^l-1) \, a x_{n+k},  
$$ 
and as both sides are a total difference this can be integrated to give 
\beq\label{genextreme} 
\sum_{j=0}^{k-1} 
\left| \begin{array}{ccc} x_{n+j} & x_{n+j+k} \\ 
x_{n+j+l}&x_{n+j+k+l} 
\end{array} \right|-
\sum_{j=0}^{l-1} a x_{n+j+k}=b, 
\eeq 
where $b$ is an integration constant. In other words, $b$ is a first integral for 
(\ref{1drecurrence}). The particular case $k=1$, given by (\ref{extreme}),  
is the ``Extreme polynomial'' family 
found from period 1 seeds in LP algebras in   \cite{laurentphenomenonsequences}, 
whose linearization was studied in \cite{honeward} for $b=0$. 

For $k>1$,  the recurrences (\ref{genextreme}) are 
not of the type (\ref{mutationrecurrence}) that can arise from periodic seeds in LP 
algebras, and none of the recurrences  (\ref{1drecurrence}) are of this type. 
Nevertheless, these recurrences turn out to have the Laurent 
property for any $k$, as a consequence of the fact that the lattice equation (\ref{lattice}) has the Laurent 
property. 

\subsection{Linearization of new lattice equation} 

Similarly to the results on linearization of the lattice equation (\ref{2dlittlepiintro}) obtained 
in \cite{kkmt}, the iterates of the new 6-point equation (\ref{lattice}), or equivalently (\ref{conslaw}),  
satisfy two types of linear relation, with coefficients that are independent of one or the other 
of the lattice variables $s,t$. 

\begin{proposition} The solutions $u_{s,t}$ of the lattice equation 
(\ref{lattice}) satisfy the linear relations
\begin{equation}\label{2dlin1}
u_{s,t+3}-(J(t+1)+1)u_{s,t+2}+(J(t)+1)u_{s,t+1}-u_{s,t}=0, 
\end{equation}
\begin{equation}\label{2dlin2}
u_{s+3,t}+A(s)u_{s+2,t}+B(s)u_{s+1,t}+C(s)u_{s,t}=0, 
\end{equation}
where $J(t)$ is independent of $s$, and $A(s),B(s)$ and $C(s)$ are independent of $t$.
\end{proposition}
\begin{proof}
Dividing both sides of  (\ref{lattice}) by $u_{s,t+1}u_{s+1,t+1}$, 
 we immediately find a quantity which is invariant under shifts in the $s$ direction, 
since the equation becomes the total difference   
$$ 
\Delta_s \, J =0, 
$$ 
where $J=J(t)$ is defined by 
\beq\label{jdefn} 
J(t):=\frac{u_{s,t+2}+u_{s,t}+a}{u_{s,t+1}}. 
\eeq
The definition of $J$ rearranges to 
give an inhomogeneous linear relation 
for $u_{s,t}$, that is 
\begin{equation}\label{inhomogenous}
u_{s,t+2}-J(t)u_{s,t+1}+u_{s,t}+a=0, 
\end{equation}
and by applying the 
difference operator $\Delta_t$ to this we are led to 
the homogeneous relation (\ref{2dlin1}). 
Since the coefficients  of the latter are fixed under  shifting $s$, we 
can write down four shifts of the relation in  the form of a matrix linear 
system, that is  
\beq\label{4by4} 
\begin{bmatrix}
u_{s,t} & u_{s,t+1} & u_{s,t+2} & u_{s,t+3} \\
u_{s+1,t} & u_{s+1,t+1} & u_{s+1,t+2} & u_{s+1,t+3} \\
u_{s+2,t} & u_{s+2,t+1} & u_{s+2,t+2} & u_{s+2,t+3} \\
u_{s+3,t} & u_{s+3,t+1} & u_{s+3,t+2} & u_{s+3,t+3} 
\end{bmatrix}
\begin{bmatrix}
-1 \\
J(t)+1 \\
-J(t+1)-1 \\
1
\end{bmatrix}
=0.
\eeq 
Thus we see that the  $4\times 4$ matrix above has determinant zero, so we may take a  vector 
$(C,B,A,1)$ in the  
left kernel, and then it is apparent that  
the entries of this vector are  invariant under  shifting  $t$. 
This kernel gives the second linear relation (\ref{2dlin2}).
\end{proof}

In the course of the proof, we observed the following result, which is 
also proved for (\ref{2dlittlepiintro}) in \cite{kkmt}. 
\begin{corollary} 
For any solution $u_{s,t}$ of 
(\ref{lattice}),  the corresponding $4\times 4$ Casorati matrix  has vanishing determinant, that is 
$$ 
\left|\begin{array}{cccc} 
u_{s,t} & u_{s,t+1} & u_{s,t+2} & u_{s,t+3} \\
u_{s+1,t} & u_{s+1,t+1} & u_{s+1,t+2} & u_{s+1,t+3} \\
u_{s+2,t} & u_{s+2,t+1} & u_{s+2,t+2} & u_{s+2,t+3} \\
u_{s+3,t} & u_{s+3,t+1} & u_{s+3,t+2} & u_{s+3,t+3} 
\end{array}  \right| =0. 
$$ 
\end{corollary}  

\subsection{Linear relations for travelling wave reductions}

Upon applying the reduction (\ref{wavereduction}), 
the linear relations (\ref{2dlin1}) and (\ref{2dlin2}) for the lattice equation 
(\ref{lattice}) reduce to linear recurrences with periodic coefficients for  
(\ref{1drecurrence}). 

\begin{proposition} 
The iterates of the equation (\ref{1drecurrence}) for the $(l,-k)$ travelling reduction 
of (\ref{lattice}) satisfy   
linear recurrence relations with periodic coefficients, given by 
\begin{equation}\label{periodiccoefficients1}
x_{n+3k}-(J_{n+k}+1)x_{n+2k}+(J_n+1)x_{n+k}-x_n=0, 
\end{equation}
\begin{equation}\label{periodiccoefficients2}
x_{n+3l}+A_nx_{n+2l}+B_nx_{n+l}+C_nx_n=0, 
\end{equation}
where the coefficient $J_n$ is periodic with period $l$, and 
$A_n,B_n,C_n$ are periodic with period $k$. 
\end{proposition}
\begin{proof}
The travelling wave reduction (\ref{wavereduction}) applied to (\ref{lattice}) gives the 
nonlinear recurrence (\ref{1drecurrence}), and  under this reduction the quantity $J(t)$ defined 
by (\ref{jdefn}) becomes
\begin{equation}\label{newJ}
J_n:=\frac{x_{n+2k}+x_n+a}{x_{n+k}}
\end{equation}
which is periodic with period $l$, while the coefficients 
$A(s), B(s)$ and $C(s)$ in (\ref{2dlin2})  become $k$-periodic quantities, denoted $A_n,B_n,C_n$.  The linear relations  (\ref{periodiccoefficients1}) and (\ref{periodiccoefficients2}) can also be constructed directly 
from the observation that $J_n$ defined by (\ref{newJ}) has period $l$. (See \cite{honeward} 
for details in the case $k=1$.) 
\end{proof}

Making appropriate adjustments compared with the case of Little Pi, we redefine 
$$ 
\Psi_n:=
\begin{bmatrix}
x_n & x_{n+k} & x_{n+2k} \\
x_{n+l} & x_{n+k+l} & x_{n+2k+l} \\
x_{n+2l} & x_{n+k+2l} & x_{n+2k+2l}
\end{bmatrix}, 
$$ 
and set 
$$ 
L_n:= 
\begin{bmatrix}
0 & 0 & 1  \\
1 & 0 & -J_n-1 \\
0 & 1 & J_{n+k}+1 
\end{bmatrix}, \qquad 
\tilde{L}_n:=
\begin{bmatrix}
0 & 1 & 0 \\
0 & 0 & 1 \\
-A_n & -B_n & -C_n 
\end{bmatrix}, 
$$ 
so that we have the linear matrix equations 
\beq\label{linmat} 
\Psi_{n+k}=\Psi_nL_n,  \qquad \Psi_{n+l}=\tilde{L}_n\Psi_n. 
\eeq 
As before, we can make the assumption $\gcd(k,l)=1$, since otherwise 
(\ref{1drecurrence}) splits into several copies of lower dimension. 

\begin{theorem}\label{extlin}
When $\gcd(k,l)=1$, 
the iterates of (\ref{1drecurrence}) satisfy the constant coefficient linear relation
\beq\label{1stK}
x_{n+3kl}-\mathcal{K}x_{n+2kl}+\mathcal{K}x_{n+kl}-x_n=0
\eeq
where $\mathcal{K}$ is the trace of the monodromy matrix, 
\[
\mathcal{K}=\tr(M_n), \qquad M_n := L_nL_{n+k}\cdots L_{n+k(l-1)}, 
\]
which is a first integral, 
as well as the linear relation 
\beq\label{2ndK} 
x_{n+(2k+1)l}-x_{n+2kl} -(\mathcal{K}-1)(x_{n+(k+1)l}-x_{n+kl})+x_{n+l}-x_n=0. 
\eeq 
\end{theorem}
\begin{proof}
From the matrix linear relations (\ref{linmat}) we have  
$$
\Psi_{n+kl}=\Psi_nM_n=\tilde{M}_n\Psi_n, 
$$ 
where the second monodromy matrix is 
$$
\tilde{M}_n:=\tilde{L}_{n+(k-1)l}\tilde{L}_{n+(k-2)l}\cdots \tilde{L}_{n+l}\tilde{L}_n. 
$$
Then since the entries of $L_n$ and $\tilde{L}_n$ have periods $l$ and $k$ 
respectively, it follows that 
$\mathcal{K}=\tr(M_n)=\tr(\tilde{M}_n)$ has period $\gcd(k,l)=1$, by assumption. 
By an analogous permutation argument to the one used in the proof of Proposition \ref{traceequality} and in 
the proof of  
Proposition \ref{traceequality2}, we have that 
$\tr(M_n)=\tr(M_n^{-1})$, and then the relation (\ref{1stK}) follows by applying 
the Cayley-Hamilton theorem to $M_n$, just as in the proof of Theorem \ref{linearrelationtheorem}. 
However, we can get a stronger result by considering (\ref{newJ}), and defining the matrices 
$$ 
\Phi_n = \begin{bmatrix} 
x_n & x_{n+k} \\ 
x_{n+l} & x_{n+k+l}  
\end{bmatrix} , 
\quad 
L_n^* = \begin{bmatrix} 
0 & -1 \\ 
1 & J_n   
\end{bmatrix} , 
\quad C^* = 
\begin{bmatrix} 
0 & 1 \\ 
0 & 1   
\end{bmatrix} , 
$$ 
which are related by the inhomogeneous equation 
$$ 
\Phi_{n+k} = \Phi_n L_n^* - a C^*. 
$$ 
Then paraphrasing the steps of the proof of Theorem 4.5 in 
\cite{honeward},  
we introduce the $2\times 2$ monodromy matrix 
$$ 
M_n^* = L_n^* L_{n+k}^*\cdots L_{n+k(l-1)}^*, 
$$ 
and find a  matrix equation of the form  
\beq\label{inhgsphi}
\Phi_{n+2kl}-\kappa\, \Phi_{n+kl}+\Phi_n = \tilde{C}_n^*,  
\eeq  
with $$\kappa = \tr( M_n^*), $$ 
where (like those of $L_n^*$) the entries of the matrix $\tilde{C}_n^*$ are periodic with period $l$.
The  top leftmost entry of (\ref{inhgsphi}) gives the equation 
\beq\label{inhgsxn}
x_{n+2kl}-\kappa\, x_{n+kl} +x_n =\tilde{J}_n, 
\eeq 
for some $l$-periodic quantity $\tilde{J}_n$, and 
if we apply the operator $\mathcal{S}^{kl}-1$ then we obtain (\ref{1stK}) together with the relation 
\beq\label{kapparel} 
\tr (M_n) = \mathcal{K}=\tr (M_n^*)+1; 
\eeq 
this also gives an independent proof that  $\tr (M_n) =\tr (M_n^{-1})$. 
However, we can instead apply the operator  $\mathcal{S}^{l}-1$ to (\ref{inhgsxn}), 
giving the homogeneous linear relation (\ref{2ndK}), which 
is of lower order than (\ref{1stK}) when $k>1$. 
\end{proof}

\begin{remark} 
For $k=1$, the quantity $\kappa=\tr(M_n^*)$ is given by the 
explicit formula 
$$ 
\kappa = \prod_{i=0}^{l-1}\left( 
1-\frac{\partial^2}{\partial J_i \partial J_{i+1}}\right) \, 
\prod_{n=0}^{l-1}J_n ,  
$$ 
and the formula for $k>1$ is obtained by a permutation of indices.  
The term of each distinct homogeneous degree in the above expression 
corresponds to a first integral of the dressing chain 
for one-dimensional Schr\"{o}dinger operators (see \cite{honeward} and 
references).   
With minor modifications, the preceding argument shows that the trace of 
the $3\times 3$ monodromy matrix  $M_n$ in (\ref{M_ndef}), defined by a product 
over matrices $L_n$ as in (\ref{lltdef}), is related via (\ref{kapparel}) 
to the trace of a $2\times 2$ monodromy matrix $M_n^*$ expressed in terms of $l$-periodic entries $J_n=K_n-1$,  
and this gives an independent proof of Proposition \ref{traceequality}, 
deriving $\mathrm{tr}(M_n) =\tr(M_n^{-1})$ 
as an equality between cyclically symmetric functions of $K_0,K_1,\ldots ,K_{l-1}$. 
\end{remark} 

\begin{remark}
The birational map  $\varphi$ in dimension $2k+l$ defined by (\ref{1drecurrence}) is measure-preserving, 
with the  volume form 
$$ 
\Omega = \frac{\rd x_0\wedge \rd x_1 \wedge \cdots \wedge \rd x_{2k+l-1}}{x_kx_{k+1}\cdots x_{k+l-1}}
$$
such that 
$$ 
\varphi^* \Omega = (-1)^l \Omega. 
$$ 
There should be $l$ independent $l$-periodic quantities $J_n$, which appear as 
coefficients in (\ref{periodiccoefficients1}), but it is unclear how many of the 
$k$-periodic quantities appearing in  
(\ref{periodiccoefficients2}) should be independent, so the question of superintegrability of 
(\ref{1drecurrence})  remains open. 
\end{remark} 

\section{Laurent property for linearizable lattice equations }\label{boatsection}

In this section we discuss the Laurent property for both of the 6-point lattice equations 
(\ref{2dlittlepiintro}) and (\ref{lattice}). 
Following van der Kamp \cite{boat}, 
we construct a family of bands of initial values,  as well as some special sets, 
that give well-defined solutions on the whole $\mathbb{Z}^2$ lattice. 
Given suitable conditions on the initial values,  linear relations with coefficients fixed in one lattice direction can be used to prove the Laurent property. We show that the band sets of initial values satisfy the necessary criteria.

\begin{definition}
Let $I$ denote a
set of initial values for a lattice equation, and let $\mathcal{L}$ denote the ring  of 
Laurent polynomials generated by $I$, that is
\[
\mathcal{L}:=\mathbb{Z}[I,I^{-1}]
\]
where $I^{-1}=\{1/u:u\in I\}$. For this $I$, given an additional set of coefficients $A$ 
appearing in the lattice equation, we say 
that a two-dimensional lattice equation satisfies the Laurent property, or is Laurent,  if
\[
u_{s,t}\in \mathcal{L}[A]
\]
for all $s,t\in\mathbb{Z}^2$. 
\end{definition}

In the equation (\ref{lattice}) there is a single coefficient $a$, so we have $A=\{a\}$, and 
we write $\mathcal{L}[a]$ for the ring of Laurent polynomials associated with an initial value set $I$.  

\subsection{Construction of band sets of initial values}\label{bandconstruction}

In \cite{boat} an algorithm is given which finds (in almost all cases) an unique  solution 
to a lattice equation on an arbitrary stencil, 
given a band of initial values $I$. 
We apply this to the $6$-point domino-shaped stencil that (\ref{lattice}) is defined on. 

First we define the lines $L_1$ and $L_2$, each with positive rational gradient, such that $L_2=L_1+(1,-2)$. 
For a given pair of lines $L_1,L_2$ related in this way, the associated band set of 
initial values $I=I(L_1,L_2)$ consists of all the 
lattice points lying between these two lines, including the points on $L_1$ but not those on $L_2$. 
An example with gradient $1/3$ is shown in Figure \ref{gradient 1/3}, 
where the points in the band set $I$ are coloured yellow. 
We will consider  each of the points on $L_1$ to be the top left corner of a $6$-point domino, hence $L_2$ will pass through the points diagonally opposite. 
By the results of \cite{boat}, taking initial values between these lines and on $L_1$ (but not on $L_2$) allows us to find 
an unique solution of (\ref{lattice}) for each choice of gradient. 
The first step is to calculate the values on $L_2$, drawn in blue, using the yellow initial values. 
We then shift the lines perpendicularly to their direction until they pass through another point of the domino, corresponding to the dashed lines in 
the figure. The new $L_2$ will pass through the next points to be calculated, drawn in red. This process is continued until we fill the whole lattice below $L_1$. We can also shift the lines in the opposite direction to fill the whole lattice above $L_1$.

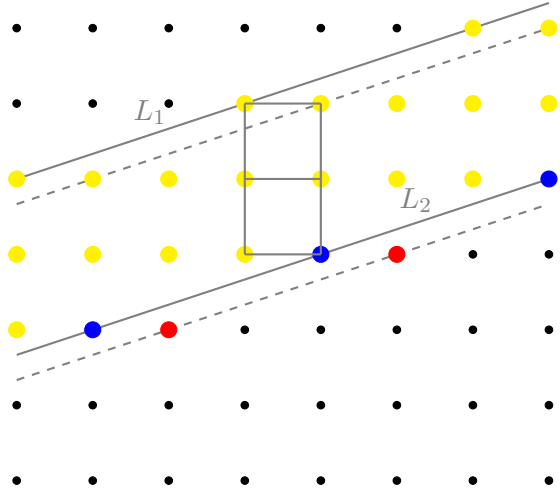
\begin{figure}

\begin{center}
	
\begin{tikzpicture}

\foreach \x in {0,1,...,7}
{
\foreach \y in {1,...,7}
{
\node[draw,circle,inner sep=1pt,fill] at (1*\x,1*\y) {};
}
}
\draw[gray, thick] (0,5) -- node[near start,above] {$L_1$} (7,7+1/3);
\draw[gray, dashed, thick] (0,5-1/3) -- (7,7);
\draw[gray, thick] (0,3-1/3) -- node[near end,above] {$L_2$}(7,5);
\draw[gray, dashed, thick] (0,3-2/3) -- (7,5-1/3);
\filldraw[blue] (4,4) circle (3pt) ;
\filldraw[blue] (1,3) circle (3pt) ;
\filldraw[blue] (7,5) circle (3pt) ;
\filldraw[red] (2,3) circle (3pt) ;
\filldraw[red] (5,4) circle (3pt) ;
\filldraw[yellow] (3,6) circle (3pt) ;
\filldraw[yellow] (4,6) circle (3pt) ;
\filldraw[yellow] (3,5) circle (3pt) ;
\filldraw[yellow] (4,5) circle (3pt) ;
\filldraw[yellow] (3,4) circle (3pt) ;
\filldraw[yellow] (0,5) circle (3pt) ;
\filldraw[yellow] (1,5) circle (3pt) ;
\filldraw[yellow] (1,4) circle (3pt) ;
\filldraw[yellow] (0,4) circle (3pt) ;
\filldraw[yellow] (0,3) circle (3pt) ;
\filldraw[yellow] (6,7) circle (3pt) ;
\filldraw[yellow] (7,7) circle (3pt) ;
\filldraw[yellow] (6,6) circle (3pt) ;
\filldraw[yellow] (7,6) circle (3pt) ;
\filldraw[yellow] (6,5) circle (3pt) ;
\filldraw[yellow] (2,4) circle (3pt) ;
\filldraw[yellow] (2,5) circle (3pt) ;
\filldraw[yellow] (5,5) circle (3pt) ;
\filldraw[yellow] (5,6) circle (3pt) ;
\draw[gray, thick] (3,4) -- (3,6);
\draw[gray, thick] (3,4) -- (4,4);
\draw[gray, thick] (4,4) -- (4,6);
\draw[gray, thick] (3,6) -- (4,6);
\draw[gray, thick] (3,5) -- (4,5);
\end{tikzpicture}

\end{center}

\caption{Initial values on the band with gradient 1/3} \label{gradient 1/3}

\end{figure} 

\subsection{The Laurent property for lattice equation (\ref{lattice})}

To prove the Laurent property we will use the linear relation (\ref{inhomogenous}), but first we must prove that the coefficients $J(t)$ belong to  the ring of Laurent polynomials.

\begin{lemma}\label{Laurentlemma}
If we have an $\tilde{s}$ such that $u_{\tilde{s},t+1}\in I$ and
\[
\{u_{\tilde{s},t},u_{\tilde{s},t+2}\}\subset \mathcal{L}[a]
\] 
then $J(t)\in \mathcal{L}[a]$.
\end{lemma}

\begin{proof}	
Since $J(t)$ defined by (\ref{jdefn}) 
is independent of $s$ we may shift it in the $s$ direction until the index value $\tilde{s}$ appears, 
and then we have 	
\[
J(t)=\frac{u_{\tilde{s},t+2}+u_{\tilde{s},t}+a}{u_{\tilde{s},t+1}}\in \mathcal{L}[a]
\]	
as required. 	
\end{proof}

\begin{theorem}\label{Laurentnesstheorem}
For a given initial value set $I$, if 	 
Lemma \ref{Laurentlemma} holds for each $t$,  and if for each $s$ there is some $\tilde{t}$ such that 
\[
\{u_{s,\tilde{t}}, u_{s,\tilde{t}+1}\}\subset \mathcal{L}[a], 
\] 
then (\ref{lattice}) has the Laurent property for this $I$.
\end{theorem}
\begin{proof}
For each $s$ we use induction on $t$ and the relation		
\[
u_{s,t+2}=J(t)u_{s,t+1}-u_{s,t}-a
\]		
noting that, by the previous lemma, $J(t)$ is in the Laurent ring. The base case for the induction is given by	
\[
u_{s,\tilde{t}+2}=J(\tilde{t})u_{s,\tilde{t}+1}-u_{s,\tilde{t}}-a.
\]
This proves Laurentness for $t>\tilde{t}+1$, and  the proof for $t<\tilde{t}$ is similar.		
\end{proof}

\begin{theorem}\label{laurentforbandtheorem}	
The Laurent property for (\ref{lattice}) holds for the band sets of initial values $I$, as described in Subsection \ref{bandconstruction}.
\end{theorem} 

\begin{proof}
To calculate $u_{s,t}$ we only have to divide by $u_{s+1,t+1}$ and vice versa, so we know all the values we calculate are Laurent polynomials 
until we have to divide at one of the blue points in Figure \ref{gradient 1/3}, and the corresponding points above $L_1$. These we mark in green in Figure \ref{gradient 1/3withgreen}. We draw $L'_2$ parallel to and below $L_2$ through the first non-green point and $L'_1$ parallel to and above $L_1$ through the last green point. Equivalently 
\[
L'_2=L_2+(-1,-1), \qquad L'_1=L_1+(-1,1), 
\]
hence $L'_2=L'_1+(1,-4)$. Since the minimal distance between $L'_1$ and $L'_2$ is $\sqrt{17}>4$ any line that intersects $I$ will intersect at least four elements of $\mathcal{L}$. For Lemma \ref{Laurentlemma} we take horizontal lines with height $t$ and see that they intersect at least four green or yellow points, at least one of which will be yellow. Hence $J(t)\in \mathcal{L}$ for all $t$. For Theorem \ref{Laurentnesstheorem} we take vertical lines for each $s$ and see that these intersect at least two green or yellow points. Hence the conditions of the 
preceding theorem hold and we have the Laurent property for these initial values.
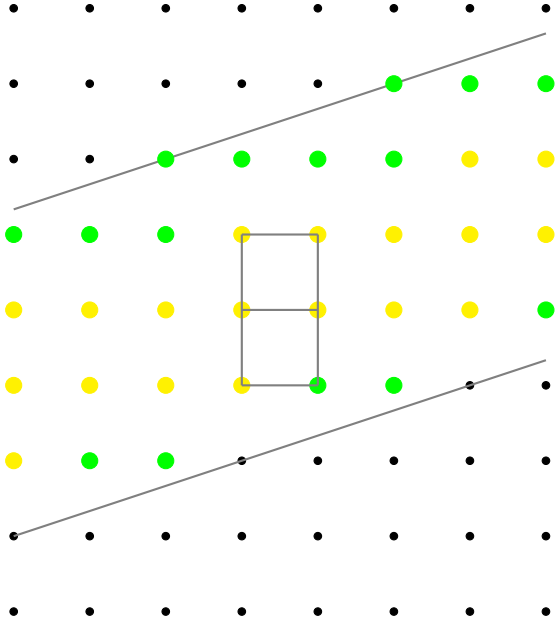
\begin{figure}
	
	\begin{center}
		
		\begin{tikzpicture}
		
		\foreach \x in {0,1,...,7}
		{
			\foreach \y in {1,...,9}
			{
				\node[draw,circle,inner sep=1pt,fill] at (1*\x,1*\y) {};
			}
		}
		\draw[gray, thick] (0,19/3) -- (7,26/3);
		\draw[gray, thick] (0,2) -- (7,13/3);
		\filldraw[green] (4,4) circle (3pt) ;
		\filldraw[green] (1,3) circle (3pt) ;
		\filldraw[green] (7,5) circle (3pt) ;
		\filldraw[green] (2,3) circle (3pt) ;
		\filldraw[green] (5,4) circle (3pt) ;
		\filldraw[green] (2,6) circle (3pt) ;
		\filldraw[green] (1,6) circle (3pt) ;
		\filldraw[green] (0,6) circle (3pt) ;
		\filldraw[green] (2,7) circle (3pt) ;
		\filldraw[green] (3,7) circle (3pt) ;
		\filldraw[green] (4,7) circle (3pt) ;
		\filldraw[green] (5,7) circle (3pt) ;
		\filldraw[green] (5,8) circle (3pt) ;
		\filldraw[green] (6,8) circle (3pt) ;
		\filldraw[green] (7,8) circle (3pt) ;
		\filldraw[yellow] (3,6) circle (3pt) ;
		\filldraw[yellow] (4,6) circle (3pt) ;
		\filldraw[yellow] (3,5) circle (3pt) ;
		\filldraw[yellow] (4,5) circle (3pt) ;
		\filldraw[yellow] (3,4) circle (3pt) ;
		\filldraw[yellow] (0,5) circle (3pt) ;
		\filldraw[yellow] (1,5) circle (3pt) ;
		\filldraw[yellow] (1,4) circle (3pt) ;
		\filldraw[yellow] (0,4) circle (3pt) ;
		\filldraw[yellow] (0,3) circle (3pt) ;
		\filldraw[yellow] (6,7) circle (3pt) ;
		\filldraw[yellow] (7,7) circle (3pt) ;
		\filldraw[yellow] (6,6) circle (3pt) ;
		\filldraw[yellow] (7,6) circle (3pt) ;
		\filldraw[yellow] (6,5) circle (3pt) ;
		\filldraw[yellow] (2,4) circle (3pt) ;
		\filldraw[yellow] (2,5) circle (3pt) ;
		\filldraw[yellow] (5,5) circle (3pt) ;
		\filldraw[yellow] (5,6) circle (3pt) ;
		\draw[gray, thick] (3,4) -- (3,6);
		\draw[gray, thick] (3,4) -- (4,4);
		\draw[gray, thick] (4,4) -- (4,6);
		\draw[gray, thick] (3,6) -- (4,6);
		\draw[gray, thick] (3,5) -- (4,5);
		\end{tikzpicture}
		
	\end{center}
	
	\caption{The points marked green are in the Laurent ring} \label{gradient 1/3withgreen}
	
\end{figure} 
\end{proof}

In the special case where the gradient is $0$ it is prescribed in \cite{boat} that we should take an extra line of initial values perpendicular to $L_1$ and $L_2$, as shown in Figure \ref{gradient0}, and this case also has the Laurent property. 
However, we note that Laurentness does not hold for all well-posed initial value problems, 
for example the yellow set shown in Figure \ref{notLaurent}. 
In this case one can see from the form of (\ref{lattice}) that to calculate the value of $u_{s,t}$ at the blue node we must divide by a polynomial (not a monomial) in the surrounding initial values.

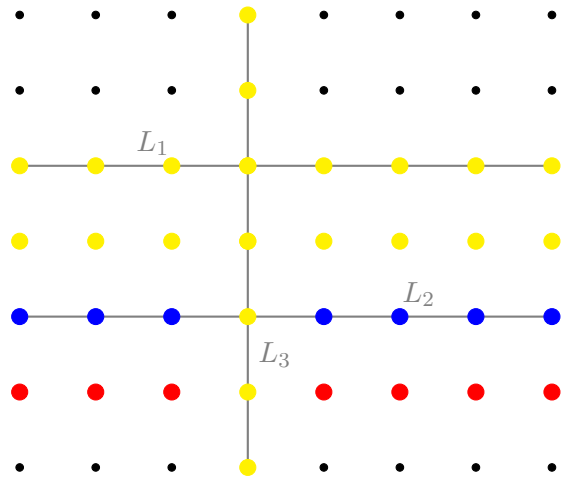
\begin{figure}
	
\begin{center}
		
\begin{tikzpicture}

\foreach \x in {0,1,...,7}
{
\foreach \y in {1,...,7}
{
\node[draw,circle,inner sep=1pt,fill] at (1*\x,1*\y) {};
}
}	
\draw[gray, thick] (0,5) -- node[near start,above] {$L_1$} (7,5);
\draw[gray, thick] (0,3) -- node[near end,above] {$L_2$}(7,3);
\draw[gray, thick] (3,1) -- node[near start,right] {$L_3$} (3,7);
\foreach \x in {0,1,...,7}
{
\filldraw[yellow] (\x,5) circle (3pt) ;
}
\foreach \x in {0,1,...,7}
{
\filldraw[yellow] (\x,4) circle (3pt) ;
}		
\foreach \y in {1,...,7}
{
\filldraw[yellow] (3,\y) circle (3pt);
}
\foreach \x in {0,1,2,4,5,6,7}
{
\filldraw[blue] (\x,3) circle (3pt) ;
}
\foreach \x in {0,1,2,4,5,6,7}
{
\filldraw[red] (\x,2) circle (3pt) ;
}
\end{tikzpicture}				
		
\end{center}
	
\caption{Initial values on the band with gradient 0 with Laurentness} \label{gradient0}
	
\end{figure}

\begin{figure}
	
	\begin{center}
		
		\begin{tikzpicture}
		
		\foreach \x in {0,1,...,7}
		{
			\foreach \y in {1,...,7}
			{
				\node[draw,circle,inner sep=1pt,fill] at (1*\x,1*\y) {};
			}
		}
		\draw[gray, thick] (0,5) -- node[near start,above] {$L_1$} (7,5);
		\draw[gray, thick] (0,3) -- node[near end,above] {$L_2$}(7,3);
		\draw[gray, thick] (3,1) -- node[near start,right] {$L_3$} (3,7);
		\foreach \x in {0,1,...,7}
		{
			\filldraw[yellow] (\x,5) circle (3pt) ;
		}
		\foreach \x in {0,1,...,7}
		{
			\filldraw[yellow] (\x,3) circle (3pt) ;
		}
		\foreach \y in {1,...,7}
		{
			\filldraw[yellow] (3,\y) circle (3pt) ;
		}
		\filldraw[blue] (4,4) circle (3pt);
		\end{tikzpicture}				
		
	\end{center}
	
	\caption{An example of initial values without the Laurent property} \label{notLaurent}
	
\end{figure}
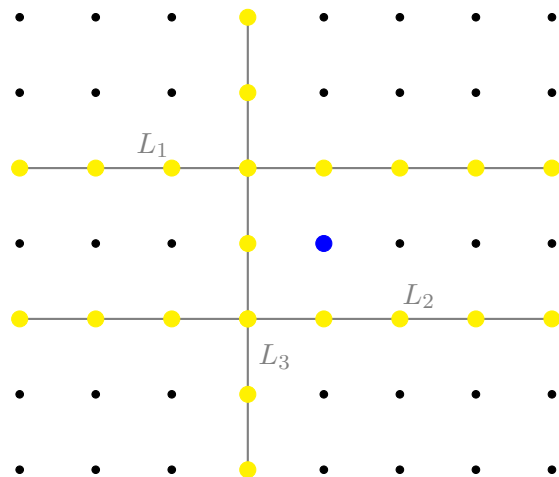 

Note that the 
Laurent property for the reductions (\ref{1drecurrence}) is easily seen from (\ref{newJ}). 
In fact, since $J_n$ has period $l$, the only initial variables that can appear in the denominator are 
$x_{k},x_{k+1},\ldots,x_{k+l-1}$. In particular, setting each of these to be $1$ will give a polynomial sequence in the remaining initial values. So we have 
\begin{corollary} The equation (\ref{1drecurrence}) has the Laurent property in the form 
$$ x_n \in 
\Z [a, x_0, \ldots, x_{k-1}, x_k^{\pm 1}, \ldots, x_{k+l-1}^{\pm 1}, x_{k+l},\ldots, x_{2k+l-1}] 
\qquad \forall n\in\Z. 
$$ 
\end{corollary}

\subsection{The Laurent property for the lattice Little Pi}

The Laurent property for the lattice equation (\ref{2dlittlepiintro}) was proved in \cite{kkmt} for 
points $(s,t)$ in the positive quadrant with the initial value set  
\[
I=\{u_{s,0},u_{s,1}, u_{0,t}:s,t\in \mathbb{N}\} 
\]
(note that we switched $s$ and $t$ compared with the original reference). 
We have drawn  the above set $I$  in yellow in Figure \ref{latticelittlepi}, extended 
to include indices $s,t$ in the whole of $\Z$.  
Again we can define the associated Laurent ring $\mathcal{L}$ (without the coefficient $a$, since we set $a\to 1$ here), 
and provide a different proof of Laurentness, similar to the proof for (\ref{lattice}). 

From Theorem 2.1 and Proposition 2.6  in \cite{kkmt}, respectively,  we have that 
\begin{equation}\label{2dlittlepilinear}
u_{s,t+6}-\beta({t+1})u_{s,t+4}+\beta(t)u_{s,t+2}-u_{s,t}=0
\end{equation}
with 
$\beta = \beta(t)$ (independent of $s$) being  given by
\begin{equation}\label{beta}
\beta(t)=\frac{1+u_{0,t}u_{0,t+3}+u_{0,t+1}u_{0,t+4}+u_{0,t+2}u_{0,t+5}}{u_{0,t+2}u_{0,t+3}}
\end{equation} 
Note that in this expression $s$ has been set to zero, but due to the fact that 
$\beta$ is $s$-independent the same formula is valid with each term $u_{0,j}$  
replaced by $u_{s,j}$ for $j=t,t+1,\ldots,t+5$.  
Similarly to the proof of Theorem \ref{laurentforbandtheorem}, we colour the values that only require division by elements of $I$  in green in Figure \ref{latticelittlepi}. Due to the shape of (\ref{2dlittlepiintro}) we end up with more green vertices than we had for (\ref{lattice}). 

\begin{proposition}\label{propositionforkkmtinitialvalues}
The lattice equation (\ref{2dlittlepiintro}) has the Laurent property for the initial values 
\[
I=\{u_{s,0},u_{s,1},u_{0,t}:s,t\in \mathbb{Z}\}.
\]
\end{proposition}
\begin{proof}
Again we fix $s$ and use induction on $t$. The induction starts with the vertical line of six values in $\mathcal{L}$, shown in yellow and green in Figure \ref{latticelittlepi}. We can see from (\ref{beta}) that, 
for this $I$, $\beta(t)\in\mathcal{L}$ for all $t$,  and we increase $t$ by solving (\ref{2dlittlepilinear}) for $u_{s,t+6}$ to obtain Laurent 
polynomials for all $t>0$, while to extend to $t<0$ we solve for  $u_{s,t}$ instead. 
\end{proof}

For the band sets of initial values we have to work harder.

\begin{lemma}\label{rotatedlittlepilemma}
For a set of initial values $I$, suppose that there is an $\tilde{s}$ such that
\[
\{u_{\tilde{s},t},u_{\tilde{s},t+1},u_{\tilde{s},t+4},u_{\tilde{s},t+5}\}\subset \mathcal{L}
\]
and
\[
\{u_{\tilde{s},t+2},u_{\tilde{s},t+3}\}\subset I;
\]
then $\beta(t)\in \mathcal{L}$. 
\end{lemma} 
\begin{proof}
In the expression for $\beta(t)$ in (\ref{beta}), 
we shift $s$  until $u_{\tilde{s},t+2}$ and $u_{\tilde{s},t+3}$ appear in the denominator, 
and the terms in the numerator belong to $\mathcal{L}$ by assumption, so the result follows. 
\end{proof}
\begin{theorem}
For a given initial set $I$, if the conditions of Lemma \ref{rotatedlittlepilemma} hold for all $t$, and if for all $s$ there is a $\tilde{t}$ such that 
\[
\{u_{s,\tilde{t}},u_{s,\tilde{t}+1},u_{s,\tilde{t}+2},u_{s,\tilde{t}+3},u_{s,\tilde{t}+4},u_{s,\tilde{t}+5}\}\subset \mathcal{L}, 
\]
then equation (\ref{2dlittlepiintro}) has the Laurent property.
\end{theorem}
\begin{proof}
The proof is the same as for Proposition (\ref{propositionforkkmtinitialvalues}).
\end{proof}

\begin{figure}
	
\begin{center}
		
\begin{tikzpicture}
		
\foreach \x in {0,1,...,8,9}
{
\foreach \y in {1,...,9,10}
{
\node[draw,circle,inner sep=1pt,fill] at (1*\x,1*\y) {};
}
}
\foreach \x in {0,1,...,9}
{
\filldraw[yellow] (\x,5) circle (3pt) ;
}
\foreach \x in {0,1,...,9}
{
\filldraw[yellow] (\x,4) circle (3pt) ;
}		
\foreach \y in {1,...,10}
{
\filldraw[yellow] (3,\y) circle (3pt);
}
\foreach \y in {1,2,3,6,7,8,9,10}
{
\filldraw[green] (2,\y) circle (3pt);
}
\foreach \y in {1,2,3,6,7,8,9,10}
{
\filldraw[green] (4,\y) circle (3pt);
}
\foreach \x in {0,1,2,4,5,6,7,8,9}
{
\filldraw[green] (\x,6) circle (3pt);
}
\foreach \x in {0,1,2,4,5,6,7,8,9}
{
\filldraw[green] (\x,7) circle (3pt);
}
\foreach \x in {0,1,2,4,5,6,7,8,9}
{
\filldraw[green] (\x,3) circle (3pt);
}
\foreach \x in {0,1,2,4,5,6,7,8,9}
{
\filldraw[green] (\x,2) circle (3pt);
}
\end{tikzpicture}					
		
\end{center}
	
\caption{The yellow dots are initial values and the green are Laurent in the initial values} \label{latticelittlepi}
	
\end{figure}
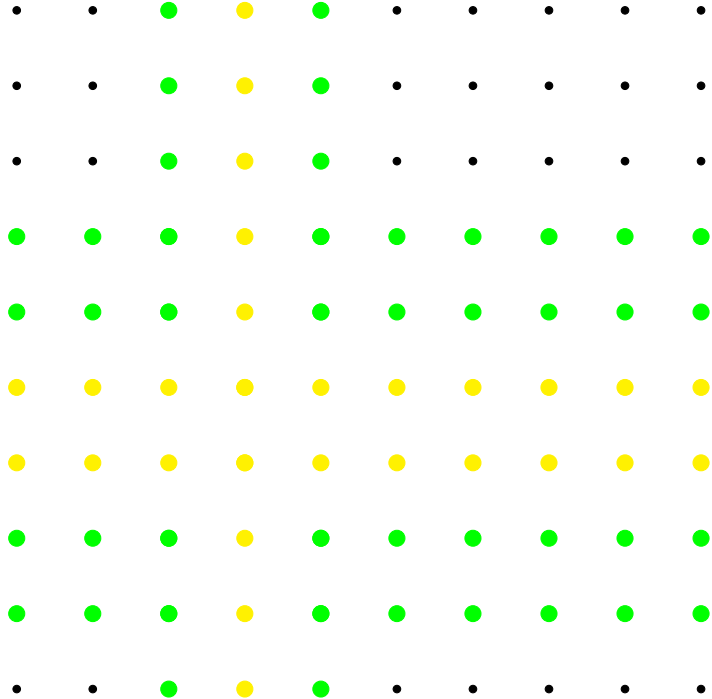

\begin{theorem}
The equation (\ref{2dlittlepiintro}) has the Laurent property if $I$ is a band of initial values.
\end{theorem}
\begin{proof}
Similarly to the proof of Theorem \ref{laurentforbandtheorem} we have
\[
L'_1=L_1+(-1,-2), \qquad L'_2=L_2+(1,-2),
\]
so $L'_2=L'_1+(3,-6)$ and the minimal distance between them is $\sqrt{45}>6$. Hence for any vertical or horizontal line intersecting the lattice we have at least 
six  consecutive values in $\mathcal{L}$, and at least two  of these will be neighbours and in $I$.
\end{proof}

\noindent 
{\bf Acknowledgements:} 
ANWH is supported by EPSRC fellowship EP/M004333/1. 
He is grateful to the School of Mathematics and Statistics, UNSW for hosting him as a 
Visiting Professorial Fellow with funding from the Distinguished Researcher Visitor Scheme, 
and thanks John Roberts and Wolfgang Schief for providing additional financial support 
for his time in Sydney. He is also grateful for the invitation from the 
organisers of the Integrable Systems workshop at the University of Sydney in December 2017, 
where he presented preliminary results on the lattice equation (\ref{lattice}). 
 JP is supported by a PhD studentship from SMSAS, University of Kent. He 
thanks UNSW and La Trobe University, Melbourne for hospitality during his visit to 
Australia in May 2018.

\end{document}